\documentclass[12pt]{article}
\usepackage{epsfig, graphicx,} 
\usepackage{amssymb, amsthm}
\usepackage{amsmath}
\usepackage{relsize}

\usepackage{calrsfs}
\DeclareMathAlphabet{\pazocal}{OMS}{zplm}{m}{n}

\newcommand{\Ib}{\pazocal{I}}
\newcommand{\Cb}{\pazocal{C}}
\newcommand{\pln}{p^+}
\newcommand{\prn}{p^-}
\newcommand{\plt}{\pln_t}
\newcommand{\prt}{\prn_t}
\newcommand{\pls}{\pln_s}
\newcommand{\prs}{\prn_s}
\newcommand{\lr}{\pm}

\usepackage{calrsfs}  

\usepackage{booktabs}
\usepackage{enumitem} 	
\usepackage{eurosym}	
\usepackage{wrapfig}	
\usepackage{float}		
\usepackage{subfig}		
\usepackage{bbm}		
\usepackage{footnote}
\usepackage{thmtools}
\usepackage[title]{appendix}
\usepackage{xcolor}
\usepackage{url}

\usepackage{hyperref}
\usepackage[authoryear]{natbib} 
\bibpunct{(}{)}{,}{a}{}{,}

\newtheorem{theorem}{Theorem}

\newtheorem{proposition}[theorem]{Proposition}

\makeatletter
\newcommand{\customlabel}[2]{%
    \phantomsection
    \protected@edef\@currentlabel{#2}%
    \label{#1}%
}
\makeatother

\newtheorem{lemma}[theorem]{Lemma}
\newtheorem{corollary}[theorem]{Corollary}
\theoremstyle{definition}

\theoremstyle{remark}

\numberwithin{theorem}{section}
\numberwithin{proposition}{section}
\numberwithin{lemma}{section}
\numberwithin{corollary}{section}
\numberwithin{definition}{section}
\numberwithin{remark}{section}
\numberwithin{example}{section}

\newcommand{\E}{\mathbb{E}}
\newcommand{\R}{\mathbb{R}}

\newcommand{\be}{\begin{equation}}
\newcommand{\en}{\end{equation}}
\newcommand{\ben}{\begin{equation*}}
\newcommand{\enn}{\end{equation*}}
\newcommand{\bea}{\begin{eqnarray}}
\newcommand{\ena}{\end{eqnarray}}

\begin{document}
	\newlength\tindent
	\setlength{\tindent}{\parindent}
	\setlength{\parindent}{0pt}
	\renewcommand{\indent}{\hspace*{\tindent}}
	
	\begin{savenotes}
		\title{
			\bf{ 
The additive Bachelier model \\
 with an application to the oil option market \\
 in the Covid period
		}}
		\author{
			Roberto Baviera$^\ddagger$  \&
			Michele Domenico Massaria$^\ddagger$ 
		}
		
		\maketitle
		
		\begin{tabular}{ll}
			$(\ddagger)$ &  Politecnico di Milano, Department of Mathematics, 32 p.zza L. da Vinci, Milano \\
		\end{tabular}
	\end{savenotes}
	
	\vspace{0.5cm}
	
	\begin{abstract}	
		\noindent
In April 2020, the Chicago Mercantile Exchange temporarily switched the pricing formula for West Texas Intermediate oil market options from the Black model to the Bachelier model. In this context, we introduce an additive Bachelier model that provides a simple closed-form solution and a good description of the implied volatility surface.

		\noindent
This new additive model exhibits several notable mathematical and financial properties. It ensures the no-arbitrage condition, a critical requirement in highly volatile markets, while also enabling a parsimonious synthesis of the volatility surface. The model features only three parameters, each with a clear financial interpretation: the volatility term structure, the vol-of-vol, and a parameter for modelling skew.

		\noindent
\textcolor{black}{Model calibration can follow a cascade procedure}: first, it accurately replicates the term structures of forwards and At-The-Money volatilities observed in the market; second, it fits the smile of the volatility surface.\ \textcolor{black}{The proposed model}
also supports efficient pricing of path-dependent exotic options via Monte Carlo simulation, using a straightforward and computationally efficient approach.

		\noindent
Overall, this model provides a robust and parsimonious description of the oil option market during the exceptionally volatile first period of the Covid-19 pandemic.

	\end{abstract}
	
	\vspace*{0.11truein}
	\vspace*{0.11truein}
	
	
	\vspace{1.5cm}
	\begin{flushleft}
		{\bf Address for correspondence:}\\
		Roberto Baviera\\
		Department of Mathematics \\
		Politecnico di Milano\\
		32 p.zza Leonardo da Vinci \\ 
		I-20133 Milano, Italy \\
		Tel. +39-02-2399 4575\\
		roberto.baviera@polimi.it
	\end{flushleft}
	
\newpage
\begin{center}
\Large\bfseries 
The additive Bachelier model with an application \\
 to the oil option market 
 in the Covid period
\end{center}
\section{Introduction}
In April 2020, the price of oil futures contracts became negative for the first time in history.
For this reason, the Chicago
Mercantile Exchange (CME) and the Intercontinental Exchange (ICE) changed their models for
oil 
from the Black to the Bachelier model for a time window 
\citep[][]{CMEGroup,ICE}.
That period 
corresponds to months when the Covid-19 pandemic spread globally, characterized by a highly turbulent oil market, and in particular for its options due to
uncertainties related to both demand and supply \citep{shaikh2021relation}. 

\smallskip

In this paper, we propose a new model that describes in a simple and parsimonious way all most liquid futures and options in the crude oil market.

There exist not many models that, on the one hand, account for negative prices in oil futures and, on the other, grant a coherent description 
for the whole volatility surface.
The Bachelier model \citep{bachelier1900jeu,davis2006} allows for negative futures prices and reproduces the term structure of At-The-Money (ATM) volatilities without considering any smile. The Bachelier model assumes that underlying prices follow a normal distribution, a significant departure from later models like \cite{black1973pricing}, which assumes a log-normal distribution. Bachelier’s framework has regained relevance in modern financial contexts, particularly in markets where asset prices can take near-zero or negative values, such as interest rate derivatives and certain commodity futures. Among possible extensions, one relevant enhancement replaces standard Brownian motions with additive processes \citep[for a comprehensive description of the properties of these processes see, e.g.,][]{Sato}. These processes, which incorporate jumps and heavy-tailed distributions, improve the model’s flexibility, enabling it to capture more accurately the real-world behaviour of asset prices, especially in volatile markets. L\'evy models --that require also stationarity-- are the most known class among additive processes (see, e.g., \citeauthor{benth2008stochastic} \citeyear{benth2008stochastic} and \citeauthor{Cont} \citeyear{Cont}): unfortunately\textcolor{black}{,} they do not have enough flexibility to offer an adequate description of some commodity markets\textcolor{black}{,} especially in highly volatile periods.

\smallskip

The proposed model extends the Bachelier model to a large class of additive processes, maintaining simple tractability and a coherent description for the whole volatility surface.
In particular, the proposed description offers a very simple calibration \textcolor{black}{procedure}, known as cascade calibration: it reproduces first the most liquid derivatives, considering then the less liquid ones. This technique was first introduced in the interest rates market \citep[see, e.g.,][]{brigo2005empirically}, where several financial products with different liquidity are involved. This elementary calibration scheme can be applied when modelling oil derivatives\textcolor{black}{,} even in the very volatile market of the first period of the Covid-19 pandemic, enabling to reproduce exactly both the term structure of futures and ATM volatility: this is a main result of this study.

The model presents other key characteristics: i) it grants a closed-form solution for plain vanilla
call/put options; ii) it provides a simple analytical expression for the implied volatility both around
the ATM and for large moneynesses, always satisfying the no-arbitrage condition; iii) it describes the whole volatility surface with only two additional parameters being very parsimonious; iv) it enables very fast simulation techniques for path-dependent exotics. 
 

\smallskip

Several are the contributions of this paper. Let us briefly summarize the three most relevant ones.
First, we introduce a new model in the Bachelier framework that is particularly easy to handle analytically. 
Second, we show that --thanks to a separability property of the implied vol-- it is possible to implement a cascade calibration, i.e.\ the model is calibrated
first on the most liquid products (discount factors, ATM forwards), 
then on the ATM options, concluding with the rest of the volatility surface.
Finally, we show an application to the West Texas Intermediate (WTI) oil market for the whole time window when the CME has considered Bachelier as the reference model (some months during the Covid-19 pandemic, when the market was very volatile)\textcolor{black}{,} obtaining excellent results.

\smallskip

The rest of the paper is organized as follows. 
In Section \ref{sec:ModelOverview}, we introduce the additive Bachelier model and its main properties.
In Section \ref{sec:Calibration}, we describe the calibration scheme.
In Section \ref{sec:Results}, we present the dataset and the main calibration results on the WTI oil market. 
Finally, we state our conclusions in Section \ref{sec:Conclusions}. 



\section{The model}
\label{sec:ModelOverview}

In this Section, we present the proposed model with
a detailed description of its several interesting properties. 
The model presents an elementary closed formula in the Fourier space for plain vanilla options (Bachelier-Lewis Formula) 
thanks to its explicit characteristic function.  
Moreover, the proposed model grants a simple description in terms of implied volatility and admits a cascade calibration, 
i.e.\ a calibration first on the most liquid financial products and then on the less liquid ones, a key property of a model among practitioners
\citep[see, e.g.,][]{brigo2005empirically}. 

\bigskip

\textcolor{black}{Let us consider a filtered probability space $(\Omega, \mathcal{F}, \mathbb{F}, \mathbb{P})$, where the filtration $\mathbb{F} = (\mathcal{F}_s)_{s \geq 0}$ is right-continuous and complete.} Let us call $F_s$ the forward valued at time $s$ and expiry at time $t$.

We \textcolor{black}{can} model the forward with
\be
F_s = F_0 + f_s \quad  \forall s \in[0,t] \;\; ,
\label{eq:ForwardDynamics}
\en
\textcolor{black}{
where $\{ f_s \}_{s \ge 0}$ is an $\mathbb{F}$-adapted process, with $f_0=0$,
within a 
class of additive processes, defined by specifying their characteristic exponent as follows
\begin{equation}
\label{eq:ChFunGeneral_s}
	\ln \phi_s(u):= \ln \mathbb{E}_0\left[e^{ i \, u \, f_s}\right]=\psi\left(i \, u \, \eta_s \, {\sigma }_s  \, \sqrt{s}   
+\frac{u^2 }{2} \, {\sigma }_s^2 \, s ; \; k_s, \; \alpha \right) \, + \,
 i\, u \, \eta_s \, {\sigma }_s \, \sqrt{s}     \;\;,
\end{equation}
with $\sigma_s,\,\eta_s,\,k_s$ continuous function of time. We indicate with $\mathbb{E}_0[\bullet]$ the expectation given the information at the value date $t_0=0$.\\
The function $ \psi(u; k, \alpha)$ is defined as
\begin{equation}
\psi (u; k, \alpha) :=
\begin{cases} 
	\displaystyle \frac{1}{k}
	\displaystyle \frac{1-\alpha}{\alpha}
	\left \{1-		\left(1+\frac{u \; k}{1-\alpha}\right)^\alpha \right \} & \mbox{if } \; 0< \alpha < 1 \\[4mm]
	\displaystyle -\frac{1}{k}
	\ln \left(1+u \; k\right)  & \mbox{if } \; \alpha = 0\;\;  \end{cases}
\label{eq:laplaceG}\,\,.
\end{equation}
and it is the Laplace exponent of a tempered stable random variable.\footnote{\textcolor{black}{Although this characteristic exponent resembles that of a subordinated process, it is useful to emphasize that $\{f_s\}_{s\geq0}$ cannot be obtained in general via subordination. The only exception is the specific case where it reduces to a L\'evy process. A similar example is discussed in \cite{azzone2025independent}.}}
The quantity $\alpha \in [0,1)$ is a hyper-parameter;
it selects the chosen model: e.g. $\alpha = 0$ is related to the VG model \citep{MadanSeneta1990} while $\alpha = 1/2$ to the NIG \citep{barndorff1997}.}

\smallskip

\textcolor{black}{An additive process is defined as a càdlàg adapted stochastic process on $\mathbb{R}$, denoted by $\{f_s\}_{s\geq0}$ with $f_0=0$ $\mathbb{P}$-a.s., which is uniquely characterized by the property of independent increments and stochastic continuity \citep[see, e.g.,][Def.14.1]{Cont}. According to Theorem 9.1 in \cite{Sato}, the distribution of an additive process at any time $s$ is infinitely divisible, and it is identified by the generating triplet $(A_s,\gamma_s, \nu_s)$ \citep[see, e.g.,][Def.8.2, p.38]{Sato}. The fundamental relationship between a system of infinitely divisible distributions and the existence of an underlying additive process is established by Theorem 9.8 in \cite{Sato}. This Theorem prescribes two essential sets of conditions on the generating triplet: firstly, monotonicity requirements to ensure the non-negativity of diffusion terms and L\'evy measures for the increments; and secondly, continuity conditions to guarantee the stochastic continuity of the process. By invoking these results, we demonstrate in Appendix \ref{app:Bachelier_formula} that there exists a family of additive processes uniquely defined by the characteristic exponent \eqref{eq:ChFunGeneral_s}.}

\bigskip

\textcolor{black}{
In this paper, we model the forward \eqref{eq:ForwardDynamics} considering a sub-case of the processes in \eqref{eq:ChFunGeneral_s}.}
The underlying process $\{ f_s \}_{s\ge 0}$ has characteristic function exponent
\begin{equation}
\label{eq:ChFun}
	\ln \phi_s(u):= \ln \mathbb{E}_0\left[e^{ i \, u \, f_s}\right]=\psi\left(i \, u \, \eta \,  {\sigma }_s \, \sqrt{ s}
+\frac{u^2 }{2} \, {\sigma }_s^2 \, s ; \; k, \; \alpha \right) \, + \,
 i\, u \, \eta \, {\sigma }_s \, \sqrt{ s}  \;\;,
\end{equation}
with $ \eta \in \R$, $k \in \R^+$ and ${\sigma }_s$ a positive continuous function of time s.t. i) ${\sigma }_s^2 \, s$ is increasing in time and ii) $\sigma_s^2\,s$ goes to zero for $s\rightarrow0^+$.\footnote{The definition \eqref{eq:ChFun} explains 
the reason why, 
we require that the volatility function $\sigma_s$ is strictly positive. 
The case $\sigma_s = 0$ for all
$s\in[0,t]$ 
describes the trivial (and non financial) situation with the 
characteristic function identically equal to $1$, i.e.\ the forward  $F_s$ identically equal to its initial value $F_0$. \textcolor{black}{Let us emphasise that additive processes with characteristic function \eqref{eq:ChFunGeneral_s} include the L\'evy processes as the stationary sub-case, as it is shown in Appendix \ref{app:Bachelier_formula}. Our model is not included in the L\'evy sub-case. We thank one reviewer for addressing this point.}}
We indicate with $\mathbf{p}$ the set of three model parameters: $\sigma_s$, $\eta$ and $k$.\footnote{\textcolor{black}{We would indicate the characteristic function \eqref{eq:ChFun} with $\phi_s(u;\mathbf{p})$, but we keep the dependence on $\mathbf{p}$ implicit when not required for clarity.}}

\smallskip

Let us anticipate that the characteristic function in \eqref{eq:ChFun} is analytic in the horizontal strip with $\Im(u)\in(-\pls,\prs)$ with $p^\pm_s>0$ and finite.\footnote{This result is proven in Proposition \ref{pr:Moment}.}

\bigskip

We call this dynamics for the forward an additive Bachelier process, because it is a natural extension of the
Bachelier model. 

\smallskip

The Bachelier model falls within this class (as shown in Appendix \ref{app:Bachelier_formula});
in this case, $f_s$ is modelled via a Brownian motion $W_s$.
Let us consider the European call option with strike $K$ and maturity $t$; let us call  
$B_0$ the discount factor between value date $t_0=0$ and maturity $t$.
Within the Bachelier model,
the price for the European call with moneyness $x := K - F_0 $ is 
\be
\label{eq:BachelierFormula}
C_b(x, t; \sigma^b_t) = B_0  \,  \E_0 \left[  F_t - K \right]^+  =
B_0  \,   \sqrt{t} \, c_b\left(  \frac{x}{\sqrt{t}}, \sigma^b_t \right) = 
B_0  \,  \sigma^b_t \, \sqrt{t} \, c_b\left(  \frac{x}{\sigma^b_t \, \sqrt{t}}, 1 \right) \, \, ,
\en
where $\sigma^b_t$ is the Bachelier volatility and $c_b$ is a normalized Bachelier call price
\be
\displaystyle c_b(y, \sigma) := - y \, \Phi \left( - \frac{y}{\sigma} \right) + \sigma \, 
\varphi \left( - \frac{y}{\sigma} \right) \,\,,
\label{eq:BasicFormula}
\en
with ${\varphi} ( \bullet )$  and ${\Phi} ( \bullet )$, respectively, the pdf and the cdf of a standard normal rv.
The main properties of this formula are summarised in Appendix \ref{app:Bachelier_formula}.

\bigskip

The following Proposition states that the forward 
(modelled via \eqref{eq:ForwardDynamics} and the characteristic function \eqref{eq:ChFun}) is an additive \textcolor{black}{and martingale} process. \textcolor{black}{The proof of the following Proposition, along with all other proofs in this paper, is provided in Appendix \ref{app:Proof_1}.}

\begin{proposition} 
\label{th:f_Additive}
The process $ \left\{f_t \right\}_{t\geq 0}$ with characteristic function \eqref{eq:ChFun}
 is additive and martingale.
\end{proposition}

This property of the proposed model is important: being additive, it is possible to simulate this process 
with an algorithm as fast as the simple Bachelier model that requires the simulation of only Gaussian rvs
\citep[see, e.g.,][]{azzone2023fast}. Thus, any structured product with a discretely \textcolor{black}{monitored} payoff can be priced and 
managed as a Bachelier model with similar characteristics in terms of simplicity and speed.
\\
Moreover, as shown in the next Section, this model provides a simple closed formula for European options, which is a generalization of the \citet{lewis2001} formula.

\subsection{European call price closed formula}

Let us consider the European call option (strike $K$, maturity $t$, moneyness $x$ and discount factor $B_0$)
\begin{equation}
\label{eq:EU_generic}
C(x, t; \mathbf{p}) =B_0\,\E_0[F_t-K]^+=B_0\,\E_0[f_t-x]^+\,\,;
\end{equation}
the proposed model allows for a very simple pricing formula.
\begin{proposition} {\rm (Call option Bachelier-Lewis formula)} 
\label{prop:Lewis}\\ 
When option's underlying is described by \eqref{eq:ForwardDynamics} and $ \left\{f_t \right\}_{t\geq 0}$ is an additive martingale process with 
characteristic function $\phi_t$, analytic in the horizontal strip $(-\plt,\prt)$ with $p^\pm_t>0$, the call option can be written as
\be
\label{eq:LewisFormula}
C(x, t; \mathbf{p}) =B_0\left(R_a+\frac{e^{x\,a}}{2\pi}\int_{-\infty}^{\infty}\left(\phi_t(\xi+ia)\frac{e^{-i\xi x}}{(i\xi-a)^2}+\frac{\mathbbm{1}_{\{a=0\}}}{\xi^2}\right)\,\mathrm{d}\xi\right)
\en
for any $a\in(-\plt,\prt)$ and
\[
R_a=
\left\{
\begin{array}{cll}
    0							&\text{if} & a\in(-\plt,0)\\[3mm]
    \displaystyle -\frac{x}{2} 	&\text{if} & a=0\\[3mm]
    -x 						&\text{if} & a\in(0,\prt)
\end{array}
\right.\;.
\]
\end{proposition}

This closed formula for the European call allows us to price plain vanilla options in a simple and fast way, thanks to its Fourier transform formulation. Let us observe that the formula \eqref{eq:LewisFormula} is generally used in practice for $a \neq 0$ in the analyticity strip $(-\plt,\prt)$. In this case, the integrand is always \textcolor{black}{analytic}: $\phi_t$ is evaluated within its strip of analyticity, and the remaining part does not contain any singularity. This closed formula can be integrated with standard quadrature techniques and even with fast Fourier transform methods \citep[see, e.g.,][]{teukolsky1992numerical} that can be extremely useful, speeding up the calibration. 
The additive Bachelier model also grants 
another call option formula.

\begin{proposition} {\rm (A second call option formula)} 
\label{prop:price} \\
When option's underlying process is described by $\{f_t \}_{t\geq 0}$ with 
characteristic function \eqref{eq:ChFun}, 
the call option is given by 
\begin{equation}
    C(x, t; \mathbf{p}) = B_0  \, \sigma_t \, \sqrt{ t}   \; 
\E \left[  c_b\left(  \frac{x}{ \sigma_t \, \sqrt{ t}} + \eta \, (G-1) , \sqrt{G} \right) \right]\,\,,
\label{eq:ModelPrice_Time}
\end{equation}
where $G$ is a positive rv with Laplace exponent $\psi(u; k, \alpha)$, defined in \eqref{eq:laplaceG}.
\end{proposition}

This alternative pricing formula looks rather similar to that of the Bachelier model.\textcolor{black}{\footnote{\textcolor{black}{\label{fn:integrals}The integral in \eqref{eq:ModelPrice_Time} is computed in the real space for $\alpha=0,1/2$, when the pdf of $G$ is known explicitly \textcolor{black}{\citep[see, e.g.,][Table 4.4, Sec.4.4.2]{Cont}}, while for the other values of $\alpha$ is computed in the Fourier space with a formula like the Bachelier-Lewis \eqref{eq:LewisFormula}.}}}
The Bachelier formula \eqref{eq:BachelierFormula} corresponds to the case when the rv $G$ is identically equal to $1$, i.e.\ 
$\psi(u; k, \alpha) = -u$ that corresponds to the limiting case $k \to 0^+$.
On the one hand, this formulation enables some important properties of the implied volatility, as discussed in the next Section, and 
on the other hand, it leads to the cascade calibration. 
The latter property -- probably the most important characteristic of the additive Bachelier model --
is presented in Section \ref{sec:Calibration}. 

\subsection{Implied volatility and its properties}

The implied volatility (hereinafter implied vol) is a common description of model option prices. 
The implied vol is defined as the value of the volatility which, 
when entered in the Bachelier option pricing model \eqref{eq:BachelierFormula}, returns a value equal to the model price \eqref{eq:LewisFormula} of that option.  Thus,
the implied vol $\mathcal{I}_t (x)$ depends on option's moneyness $x$ and 
maturity $t$ and it is obtained by solving the implied vol (IV) 
equation
\begin{equation}
    C_b(x, t; \mathcal{I}_t (x)) = C(x, t; \mathbf{p}) \;\; .
\label{eq:IV eq}
\end{equation}

\begin{proposition} {\rm (Existence and uniqueness of the implied vol)}

Given $\{f_t \}_{t\ge 0}$, \textcolor{black}{for any} $ {t} \in \R^+\setminus\{0\}$ and ${x} \in \R$, 
it exists a unique ${\cal I}_t(x) \in \R^+$ solution of the IV equation \eqref{eq:IV eq}.
\label{pr:Existence}
\end{proposition}

In the literature, some authors
consider the 
moneyness degree, the moneyness divided by the square root of the time-to-maturity (ttm), 
showing that in some cases it is a natural model description  \citep{carr2003finite,medvedev}.
For the additive Bachelier, the implied vol ${\cal I}_t(x)$, written in terms of the (normalized) moneyness degree
\begin{equation}
    y := \frac{x}{\sigma_t  \, \sqrt{t}} \; ,
\label{eq:mondegree}
\end{equation}
has a very simple expression, as stated in the following Proposition.

\begin{proposition} {\rm (Separability of the implied vol)}
\label{prop:IyTimeIndep} 

\textcolor{black}{For any $y\in\mathbb{R}$, $t\in\mathbb{R}^+\setminus\{0\}$,} the implied vol \textcolor{black}{${\cal I}_t(\bullet)$} is a separable function of $y$ \eqref{eq:mondegree} and $t$
\textcolor{black}{\begin{equation}
    \label{eq:ImpVolSeparability}
{\cal I}_t\left(y\,\sigma_t\,\sqrt{t}\right) = \sigma_t \, I(y) \; ,
\end{equation}}
\textcolor{black}{where ${I}(\bullet)$ is the solution of the time independent equation}
\begin{equation}
    \label{eq:MainIV}
 {c_b\left( y, {I}(y)\right)}  =  \E \left[ c_b\left(y+  \eta \, (G -1) , \sqrt{G}\right) \right]     \; ,
\end{equation}
\textcolor{black}{where $G$ is the rv in \eqref{eq:ModelPrice_Time}.}
\end{proposition}

Let us notice that the above property is a fundamental characteristic of this model. The possibility to separate the dependency due to the maturity in the implied volatility will allow us to develop a calibration scheme that reproduces exactly ATM vols as described in Section \ref{sec:Calibration}. By abuse of language, whenever there is no ambiguity, we also call $I(y)$ implied vol or volatility smile, 
since it describes the volatility smile of the implied vol.\textcolor{black}{\footnote{\textcolor{black}{As noted above regarding the characteristic function $\phi_s(u)$, the implied vol would be formally denoted by $I(y;\eta,k)$. However, to avoid cluttering the notation, we write $I(y)$, keeping the dependence on $\eta$ and $k$ implicit when it is clear from the context. In instances where only one of the two parameters is relevant, e.g.\ $\eta$, we denote the volatility smile as $I(y;\eta)$.}}}

\smallskip

The unique implied vol ${ I}({ y})$ can be obtained by solving the very simple IV equation (\ref{eq:MainIV}), and it presents several interesting properties.
In particular: i) it is a very regular function, being ${\cal C}^2 ( \R)$,  ii) its symmetry is regulated by the parameter $\eta$
and iii) its expression is known both for small and large (normalized) moneyness degrees.
These properties are described in the following Propositions.

\begin{proposition} {\rm  (Regularity of ${ I}
$)}

For any $\eta\in\R$ and $ k \in \R^+$,  ${ I}({ y})$ is ${\cal C}^2 ( \R)$.
\label{pr:Regularity}
\end{proposition}

It is possible to observe that for $k$ equal to zero, the unique solution of the IV equation (\ref{eq:MainIV}) is ${ I}({ y}) = 1$ $\forall { y} \in \R$, corresponding to the Bachelier model with time dependent deterministic volatility \textcolor{black}{in \eqref{eq:BachelierFormula}}.

\smallskip


\begin{proposition} {\rm (Symmetry of ${ I}
$) }
\label{pr:Symmetry}

\textcolor{black}{For any $y\in\mathbb{R}$, $\eta\in\mathbb{R}$ and $k\in\mathbb{R}^+$, }
${ I}({ y}; \eta) = { I}(-{ y}; -\eta) $.
\end{proposition}

\begin{corollary} {\rm (${ I}$  even function of ${ y}$) }\\
\textcolor{black}{For any $k\in\mathbb{R}^+$,} \textcolor{black}{${ I}(y;\eta)$}  is an even function of ${ y}$, iff $\eta=0$.
\label{cor:Symmetry}
\end{corollary}

\smallskip

The above properties suggest a comment on the parsimony of the proposed model.
\textcolor{black}{Not only does the additive Bachelier have few parameters} --just three-- being able to model the volatility surface.
They also describe three main distinct aspects of the surface, as we'll discuss in detail in Section \ref{sec:Calibration}, $\sigma_t$ models the term structure of the ATM volatility, i.e.\ it controls the level of the volatility surface.
Proposition \ref{pr:Symmetry} and Corollary \ref{cor:Symmetry} state that $\eta$ models the volatility skew, controlling the symmetry of the volatility, and for $\eta=0$ describes the surface with a perfectly symmetric smile \textcolor{black}{wrt} the ATM.
Finally, $k$ models the vol-of-vol, i.e.\ the convexity of the surface, being the limit $k \to 0^+$ the classical Bachelier model with no smile.

\bigskip

\begin{proposition} {\rm (Implied volatility  ${ I}$ for small $|y|$, i.e.\ close to the ATM)} \label{pr:ATMVolSkew}\\
\textcolor{black}{For any $\eta\in\mathbb{R}$, $k\in\mathbb{R}^+$,} the implied volatility  ${ I}$ around the ATM \textcolor{black}{(i.e., for $y \to 0$)} is
\[
{ I}({ y}) = { I}_0 +  { I}'_0 \, { y} + \frac{1}{2} { I}''_0 \, { y^2} +  o({ y^2}) 
\]
where \textcolor{black}{
\be
\left\{
\begin{array} {ccrl}
 { I}_0 :=&I(0) &=& 
\sqrt{2 \pi}  \;\,\E \left[  c_b\left(  {\eta} \, (G - 1) , \sqrt{G}\right)  \right] \\[5mm]
 { I}'_0 :=
 & \displaystyle\left. \frac{\partial {{I}} (y) }{\partial y} \right|_{y=0}&=& \displaystyle - \sqrt{\frac{\pi}{2}}  \displaystyle \;\,\E \left[ {\rm erf} \left(  \frac{  {\eta}}{\sqrt{2}} \, \frac{  1 - G }{  \sqrt{G}} \right) \right] \\[5mm]
{ I}''_0 :=
 & \displaystyle\left. \frac{\partial^2 {{I}} (y) }{\partial y^2} \right|_{y=0}&=& \displaystyle \sqrt{2 \pi} \; \displaystyle \;\,
 \E \left[ \frac{1}{  \sqrt{G}} \, \varphi \left(  {\eta} \, \frac{  1 - G }{  \sqrt{G}} \right) \right] - \frac{1}{I_0}
\label{eq:ATM coeff}
\end{array}
\right.  \; ,
\en}\\
with $I'_0$ an odd function in $\eta$ and $G$ is a positive rv with Laplace exponent $\psi(u; k, \alpha)$, defined in \eqref{eq:laplaceG}.
\end{proposition}

\smallskip

These results are extremely relevant for both ATM volatility and volatility skew.\footnote{\textcolor{black}{The integrals in \eqref{eq:ATM coeff} are computed in the real space for $\alpha=0,1/2$, as in footnote \ref{fn:integrals}.}}

The ATM vol, using implied vol separability \eqref{eq:ImpVolSeparability}, is
\textcolor{black}{
\be
\sigma^{ATM}_t := \mathcal{I}_t (x=0) = \sigma_t \, I_0
\label{eq:ATMvol}
\en}\\
i.e.\ the ATM volatility is proportional to $\sigma_t$, with the proportionality coefficient equal 
to ${ I}_0$ in \eqref{eq:ATM coeff} that depends only on $\eta$ and $k$.
Moreover, the quantity $ { I}_0$ is strictly positive, because it is the expectation of a positive function.\\
This fact will play a key role in the model calibration scheme, as discussed in Section \ref{sec:Calibration}.

\smallskip
 
The volatility skew (often simply named skew) is defined as \citep[see, e.g.,][]{gatheral2011volatility}
\[
skew := \left. \frac{\partial  \mathcal{I}_t (x) }{\partial x} \right|_{x=0} \; .
\]
It is proportional to $I'_0$. Indeed, the skew is equal to
\begin{equation}
    \left. \frac{\partial  \mathcal{I}_t (x) }{\partial x}  \right|_{x=0} = 
\sigma_t \, \left. \frac{\partial {{I}} (y) }{\partial y} \right|_{y=0} \frac{\partial y}{\partial x} = \frac{1}{\sqrt{t}} \; I'_0 \;\; .
\label{eq:modelSkew}
\end{equation}
Thus, we obtain a diverging skew for small times, a stylized fact of volatility surfaces.

\bigskip

The following result extends \cite{lee2004moment} formula, providing the characterisation of the asymptotic behaviour of the implied vol in the additive Bachelier model. Let us notice that the asymptotic behaviour is considered for the limit in Proposition \ref{pr:Moment}, while \cite{lee2004moment} considers the $\lim\sup$.
\begin{proposition} {\rm  (Implied volatility  ${ I}$  for large ${|y|}$) }
\label{pr:Moment}

\textcolor{black}{For any $\eta\in\mathbb{R}$, $k\in\mathbb{R}^+$}, the asymptotic behaviour of ${ I}({ y})$ for large $|{y}|$ is
\[
\lim_{{ y} \to \pm \infty} \frac{{ I}({ y})^2}{|{ y}|} =  \frac{1}{2 \; p^\lr} 
\]
with 
\[
p^\pm =  \lr \eta + \sqrt{ {\eta}^2 + 2 \, \frac{1-\alpha}{k}  }  \: .
\]
\end{proposition}

\section{Calibration scheme}
\label{sec:Calibration}

Model parameters have a clear financial interpretation and -- as already anticipated in the Introduction --
they can be calibrated following a cascade calibration, i.e.\ a calibration procedure that includes the derivatives in several stages according to their liquidity.

The cascade calibration in option markets is the procedure desired by any market maker when calibrating a model.
He/she would like to: first, reproduce the most liquid products %
(discount-rates and forwards/futures), then calibrate the most liquid options 
(ATM calls/puts) and finally fit the volatility surface for options that are either In-The-Money (ITM) or Out-The-Money (OTM).  
The idea is to get a more precise description of derivative contracts that are very liquid, thus presenting a narrow bid-ask spread
as futures or ATM options, leaving only at a final stage 
the model calibration of the other options in the volatility surface, 
where trades are rarer and with larger bid-ask spreads.

The additive Bachelier allows for a three-stage calibration scheme:
i) first, discount-rates and ATM forwards; ii) then, ATM options and iii) finally, all remaining traded options.
As we'll show in the next Section, such a calibration scheme is crucial in a very volatile market as the
oil derivative market during the first Covid months.

\bigskip

First, the model can be calibrated exactly to the ATM forwards from market data.
In particular, we follow the approach in \citet{azzone2021synthetic}, which allows us to extract from option prices 
precisely --with an elementary regression-- the term-structure used by option market makers of both (synthetic) forwards and
discount-rates. 

\smallskip

Second, we compute from option prices the term-structure of the ATM vol $\sigma^{ATM}_t$.
We select the parameter 
$\sigma_t$ proportional to the observed ATM vol $\sigma^{ATM}_t$.
This calibration allows us to reproduce exactly market ATM prices, in particular if we select the 
proportionality constant equal to $I_0$, as shown 
in \eqref{eq:ATMvol}.

\smallskip
Finally, we desire to calibrate the last two parameters of the additive Bachelier ($\eta$ and $k$) on the volatility surface.
We compute a normalized  moneyness ${\chi}$ for the quoted options using the 
observed ATM forwards, ATM vol, maturities, and strikes 
\begin{equation}
\label{eq:mon_norm_k}
{\chi} :=  \frac{x}{\sigma^{ATM}_t \,\sqrt{t} }  \;\; .    
\end{equation}

This quantity is equivalent to the moneyness degree \eqref{eq:mondegree},
but differs for the constant ${ I}_0$, having that two relations hold
\[
\left\{
\begin{array}{lcl}
\sigma^{ATM}_t &=& { I}_0  \, \sigma_t \\
y &=& { I}_0  \,  {\chi}
\end{array} 
\right. \;\; .
\]

We have already discussed the separability property 
\eqref{eq:ImpVolSeparability}
of the implied vol, which enables us to separate the dependence due to the term-structure (controlled by the parameter $\sigma_t$)
and the one deriving from the moneyness.
A similar property can be rewritten in terms of $t$ and $\chi$.

The implied vol, for any option with maturity $t$,  can be indicated in terms of ${\chi}$ as
\be
{\cal I}_t (x) = \sigma^{ATM}_t \, \Ib({\chi} ) \, 
\label{eq:ImpVolSeparabilityATM}
\en
where
\[
\Ib({\chi}) := \frac{I({\chi} \, { I}_0)}{I_0} \; 
\]
does not depend on the maturity $t$, as stated in the following Proposition.


\begin{proposition}\label{pr:altEqCalibration}
Model IV equation  \eqref{eq:MainIV} is equivalent to
\be
\label{eq:IVnorm}
 {c_{b}\left ({\chi} ,  \Ib({\chi} )\right)} =  \Cb ({\chi}; \eta, k)   \; ,
\en
where
\[
\Cb ({\chi}; \eta, k)   :=
\frac{1}{{ I}_0}
\E \left[ c_b\left({\chi} \, { I}_0 +  \eta \, (G -1) , \sqrt{G}\right) \right] \;\; ,
\]
and the implied vol $\Ib({\chi})$, unique solution of \eqref{eq:IVnorm}, does not depend on $t$.
\end{proposition}

Let us observe that the right hand side of the IV equation (\ref{eq:IVnorm}), $ \Cb ({\chi}; \eta, k) $, depends only on the parameters $\eta$ and $k$. Thus, $\Ib({\chi})$ is 
function of only these two parameters, while it does not depend on $\sigma_t$.
This separability property is important since it enables a three-stage calibration scheme.
%
%
\begin{enumerate}
\item Calibrate ATM forwards and discount-rates on market data;  
\item Identify $\sigma^{ATM}_t$ on quoted options and choose $\sigma_t$ equal to $\sigma^{ATM}_t/I_0$ for every maturity $t$;
\item Calibrate the two parameters $\eta, k$ on OTM option prices.
\end{enumerate}

\textcolor{black}{In the marketplace, it is now standard to obtain ATM forwards and the discount factors via a simple linear regression of call/put prices, following the methodology described in \cite{azzone2021synthetic}. Specifically, for each maturity, a linear regression of the difference between observed call and put prices against their strike prices is performed. This allows us to identify the market-implied discount factor and the ATM forward price simultaneously. This step ensures that the subsequent calibration of the volatility surface is consistent (and coincides exactly) with the market-implied term structures of forwards.}

\smallskip

As already mentioned, the first two stages are able to reproduce exactly both the term structures of forwards and ATM vol; once we impose the condition $\sigma_t=\sigma_t^{ATM}/I_0$, model price $C(x,t;\mathbf{p})$ is function only of the parameters $\eta$ and $k$ \begin{equation*}
\displaystyle C^{mod}(x,t;\eta,k):=C\left(x,t;\sigma_t=\frac{\sigma_t^{ATM}}{I_0},\eta,k\right)\,=B_0\,\sigma_t^{ATM}\,\sqrt{t}\,\Cb ({\chi}; \eta, k) \,\,.
\end{equation*}
The last stage leverages on the properties of the implied vol surface (in particular, the separability 
\eqref{eq:ImpVolSeparabilityATM}
of the implied vol discussed above) and requires to 
minimize the $L^2$-distance between the quoted market prices $C^{mkt}(x_i,t_i)$ and the corresponding model prices
\begin{equation}
\label{eq:distance}
\begin{aligned}
    d(\eta, k) :&= \sum_{i}\left[ C^{mkt} (x_i,t_i) - C^{mod}(x_i,t_i; \eta, k)\right]^2=\\&\textcolor{black}{=\sum_{i}\left[ C^{mkt} (\chi_i\,\sigma_{t_i}^{ATM}\,\sqrt{t_i},t_i) - B_0\,\sigma_{t_i}^{ATM}\,\sqrt{t_i}\,\Cb ({\chi}_i; \eta, k)\right]^2} \; ,
\end{aligned}
\end{equation}
selecting the optimal values of $\eta$ and 
$k$.\footnote{The model is calibrated on both calls and puts, considering only the OTM options, thanks to the call-put parity.} \textcolor{black}{Let us emphasize, that the evaluation of the distance $d(\eta,k)$ in \eqref{eq:distance} is particularly fast for this model: it requires the computation of only $\Cb(\chi,\eta,k)$ for different values of $\chi_i$ that can be obtained applying only one fast Fourier transform.}\\
Finally, the normalized implied vol $\Ib({\chi} )$ can be obtained via the implied vol equation \eqref{eq:IVnorm}. 
The calibrated implied vol is 
\[
{\cal I}_t\left( {x} \right) = \sigma^{ATM}_t \; \Ib \left( \frac{x}{\sigma^{ATM}_t \, \sqrt{t} } \right) \; .
\]

We observe that this calibration method reproduces exactly the market ATM volatility. 
In fact, ATM $\Cb(0; \eta, k)$ is equal to $(2 \, \pi)^{-1/2}$ for every $\eta$ and $k$, 
thus $\Ib \left( 0 \right) =1$.


\bigskip


A property that describes well the implied vol surface in some markets is sticky delta.
Sticky delta describes the situation where the implied volatility remains unchanged (i.e.\ it sticks) with the moneyness when the underlying moves.

\begin{lemma} {\rm (Stickiness of the implied vol)}
\label{lem:sticky}

Additive Bachelier implied vol is sticky delta wrt $x$.
\end{lemma}

This model property is a great advantage from the practitioners' perspective: 
hedges derived from this model remain stable if the market moves according to the sticky delta rule, i.e.\ both the volatility surface and prices (normalized with the discount) as functions of moneyness remain constant over time. This is another difference wrt exponential models, where only the prices normalized by the ATM forward remain constant.

In the next Section, we'll observe that while the model fits exactly the most liquid products (discount-rate, forward, and ATM vol), 
the other parameters ($\eta$ and $k$) remain relatively constant even in the high-volatile oil market of the first 
Covid months. Furthermore, to show this stability of prices with $\eta$ and $k$, we'll also keep the same value of $\eta$ and $k$ even for the next day, obtaining prices that are almost indistinguishable from the calibrated prices with $\eta$ and $k$ optimally selected. We'll also repeat the valuation even keeping the same parameters ($\eta$ and $k$) for an entire week: observed price differences remain generally within $1$-$2$ cents, well within bid-ask prices.\footnote{A trader that keeps hedged his/her portfolio, in practice hedges exposures generated (according to a given model) 
by model parameters that 
--ideally-- correspond to different risks (e.g., underlying, discount-rate, ATM vol, skew, etc...).
If these parameters remain relatively constant over time when market conditions do not change significantly,
he/she avoids paying transaction costs on trades due to exposures generated by an inadequate model description.}

\section{Calibration results}\label{sec:Results}
In this Section, we present the results of the calibration on market data of the proposed model, comparing the results with 
three natural benchmarks.

\smallskip

In Section \ref{sec:data_description}, we describe the data used in the analysis; in Sections \ref{sec:all} and \ref{sec:one_date}, we show the results of the calibration. In particular, in Section \ref{sec:all}, we show the results for all the value dates considered in this analysis. In Section \ref{sec:one_date}, we focus on a specific value date in the dataset, the $29^\mathrm{th}$ of April 2020 (one week after the \textcolor{black}{announcement} from the CME, that changed \textcolor{black}{the} reference model into Bachelier), comparing our results with market data and three benchmarks.

\subsection{Dataset}\label{sec:data_description}

The dataset is composed of daily market prices of call and put options on 
the WTI oil futures traded on the CME. 
We consider the whole period when the CME has switched the options pricing and valuation model to Bachelier \citep{CMEGroup}.
The time window of interest goes from the $23^\mathrm{rd}$ of April 2020 to the $31^\mathrm{st}$ of August 2020. 
Data are present \textcolor{black}{on} all weekdays \textcolor{black}{except} two public holidays in the USA: the $25^{\mathrm{th}}$ of May 2020 (Memorial Day) and the $3^{\mathrm{rd}}$ of July 2020 (Independence Day \textcolor{black}{holiday}), when no options were traded.
Daily prices are obtained from Bloomberg. The options are quoted for a grid of strikes that goes from 2.5 \$ to 245 \$. The grid size is equal to 0.5 \$ from 5 \$ to 146.5 \$, it's equal to 2.5 \$ from 147.5 \$ to 195 \$, and it's 5 \$ from 195 \$ to 245 \$. 
We consider all available liquid options' expiry: from June 2020 (JUN20) to December 2022 (DEC22). 
We have considered as ``risk-free" rate the Effective Federal Funds Rate OIS curve in the time window of interest up to the three years expiry.
In Table \ref{tab:refDates_strikes}, we report a summary statistics.


\begin{table}[H]
    \centering
\small{
    \begin{tabular}{|l|c|c|c|}
    \toprule
         & Min & Max & \# 
        \\
        \bottomrule
        \toprule
        Value Dates & $23^{\mathrm{rd}}$ of April 2020 & $31^{\mathrm{st}}$ of August 2020 & 91 \\
        Strikes & 2.5 \$ & 245 \$ & 314 \\
        Expiries & JUN20 & DEC22 & 9 \\
        \bottomrule
    \end{tabular}}
    \caption{\small{Value dates, grid of strikes and option expiries in the dataset. We report the number of value dates, strikes, and expiries in the dataset.}}
    \label{tab:refDates_strikes}
\end{table}

From this dataset, we first replicate discount-rates and ATM forwards from option prices, as described in \cite{azzone2021synthetic}. To do so, we need the couples of call and put options traded for the same strike and expiry, for each value date. In Table \ref{tab:put_call_couples}, we show the number of couples of options available for each value date in the period of the analysis\textcolor{black}{. We} report mean, median, standard deviation (std), and quantiles (q) 5\%, 95\%.

\begin{table}[H]
    \centering
\small{
    \begin{tabular}{|l|c|c|c|c|c|c|c|c|c|}
        \toprule
         & JUN20 & SEP20 & DEC20 & MAR21 & JUN21 & SEP21 & DEC21 & JUN22 & DEC22 \\    
        \bottomrule
        \toprule
        Mean & 148 & 104 & 100 & 19 & 22 & 4 & 24 & 3 & 7 \\
        Median & 148 & 107 & 114 & 18 & 21 & 4 & 24 & 3 & 8 \\
        Std & 0 & 17 & 30 & 7 & 5 & 2 & 5 & 1 & 2 \\
        $q_{0.05}$ & 148 & 96 & 28 & 11 & 17 & 1 & 16 & 0 & 4 \\
        $q_{0.95}$ & 148 & 114 & 117 & 31 & 30 & 6 & 29 & 3 & 9 \\
        \bottomrule
    \end{tabular}}
    \caption{\small{We report the number of call-put options' couples traded for each maturity and value date: mean, median, standard deviation (std), and quantiles (q) 5\%, 95\%\textcolor{black}{,} considering each value date.}}
    \label{tab:put_call_couples}
\end{table}

After having taken into account the ATM vol, we desire to calibrate the two parameters of the additive model from OTM options. We consider OTM call options ($K>F_0$) and OTM put options ($K<F_0$) in the range of moneyness $x\in[-30\,\$,30\,\$]$.\\
In Table \ref{tab:desc}, we provide the descriptive statistics of the number of available OTM options, for all the strikes and value dates\textcolor{black}{. We} report mean, median, standard deviation (std), and quantiles (q) 5\%, 95\%.

\begin{table}[H]
    \centering
\small{
    \begin{tabular}{|l|c|c|c|c|c|c|c|c|c|}
    \toprule
         \# OTM opt. & JUN20 & SEP20 & DEC20 & MAR21 & JUN21 & SEP21 & DEC21 & JUN22 & DEC22 \\     
         \bottomrule
         \toprule
         Mean & 92 & 116 & 105 & 58 & 49 & 13 & 49 & 10 & 24 \\
         Median & 92 & 120 & 119 & 51 & 46 & 13 & 54 & 9 & 26 \\
         Std & 9 & 7 & 27 & 25 & 7 & 3 & 9 & 1 & 6 \\
         $q_{0.05}$ & 76 & 101 & 30 & 26 & 42 & 7 & 26 & 8 & 7 \\
         $q_{0.95}$ & 103 & 120 & 120 & 94 & 63 & 17 & 56 & 11 & 28 \\
         \bottomrule
    \end{tabular}}
    \caption{\small{Number of OTM options considered in the volatility surface calibration. We report the number of OTM \textcolor{black}{calls} ($K>F_0$) and \textcolor{black}{puts} ($K<F_0$) in the moneyness range $[-30\,\$,30\,\$]$: mean, median, standard deviation (std), and quantiles (q) 5\%, 95\% considering each value date.}}
    \label{tab:desc}
\end{table}

We observe from these Tables that the market in the period of interest (from April 2020 to August 2020) was relatively illiquid. We'll notice in the following that, despite these market conditions, our model is able to describe the key features of this market adequately, with stable results across the value dates in the period of interest.

\subsection{A glimpse on all value dates}\label{sec:all}

In this Section, we present the results for all value dates in the dataset; in particular, we show the spread of 
options' discount-rates on the OIS rates, and the parameters $\eta$ and $k$ of the additive model; we present their evolution across the different value dates. \textcolor{black}{In this Section, we consider the model with the choice of $\alpha=1/2$ (results for other values of $\alpha$ are similar).}

\begin{figure}[H]
    \centering
    \includegraphics[width=1\linewidth]{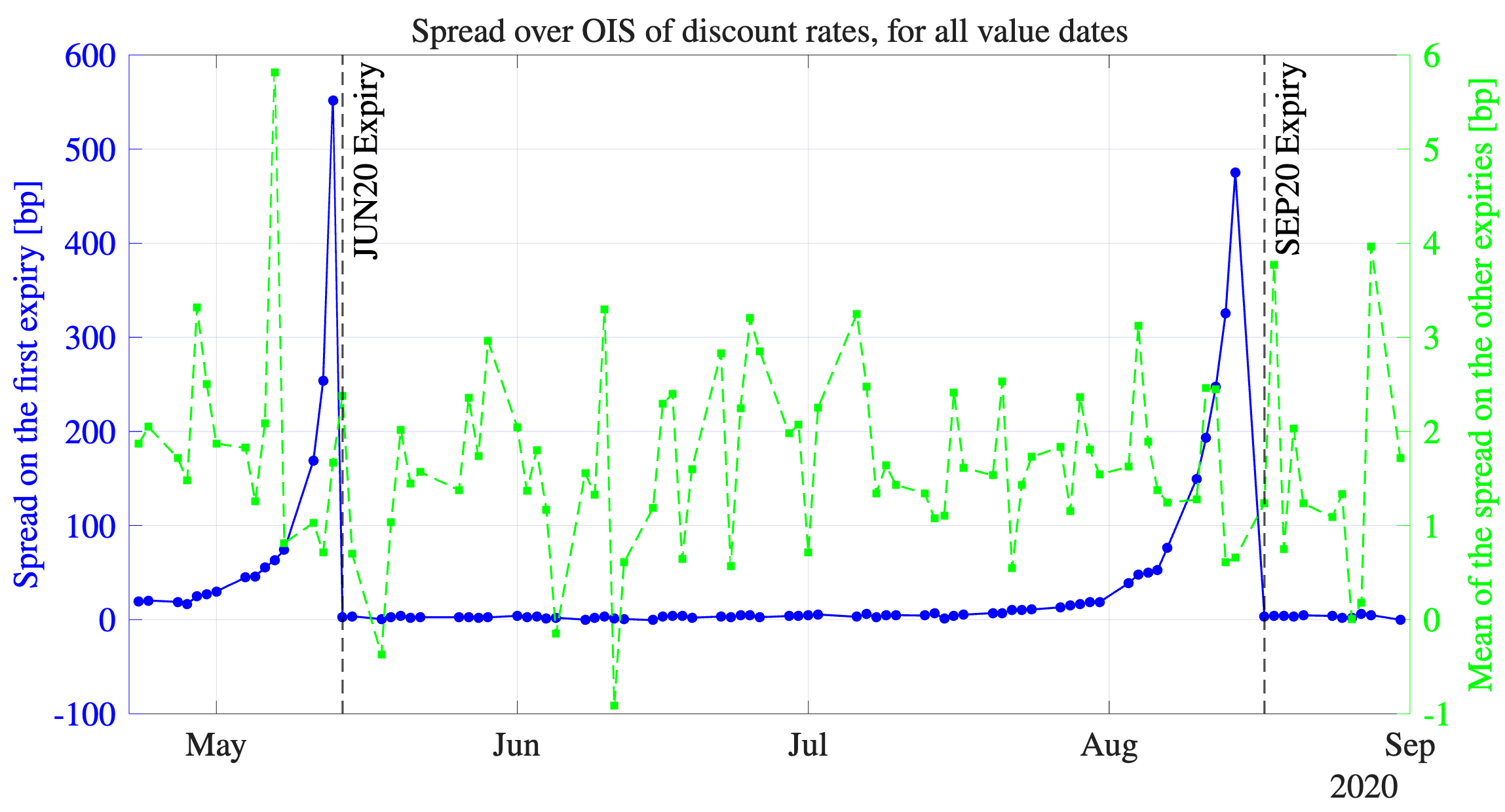}
    \caption{\small{Spread of the discount-rates (from option prices) on the OIS curve, for all value dates considered in the analysis. The vertical lines correspond to the expiries of the first two options. In blue (dots and continuous line, left axis), we can see the spread corresponding to the first maturity: we observe that it rapidly increases when approaching options' first expiry, reaching some hundreds of basis points. In green (squares and dashed line, right axis), we see the average of the spread corresponding to the other expiries: we can observe that it is of \textcolor{black}{a} few basis points (an average of $2$ bps), i.e.\ its order of magnitude is much smaller than the spread for the first expiry.}}
    \label{fig:spread_all}
\end{figure}

Figure \ref{fig:spread_all} shows the spread of the discount-rates derived from option prices on the ``risk-free" rates (OIS) for different value dates. In the blue line (dots and continuous line, left axis), we can see the spread corresponding to the first maturity: this quantity shows a rapid increase when approaching the first expiry. The average of the corresponding spread on all the other expiries is plotted in green (squares and dashed line, right axis). We observe that the latter remains of the order of \textcolor{black}{a} few basis points (on average around $2$ bps), i.e.\ negligible for all practical purposes in the volatile market of the first Covid period, and it is much smaller than the spread on the first expiry (even of $2$ orders of magnitude), that reaches some hundreds of basis points. 
The inability to reproduce even OIS rates is a clear indicator of market distress. 
The observed behaviour of option discount-rates for very short-term maturities highlights the extremely tough market conditions that the WTI oil market experienced during the Covid period.

\smallskip

For this reason, the options on the first expiry are not considered in the calibration procedure
when the discount-rates are significantly different from the ``risk-free" rates; in particular, the adopted criterion is to consider the first expiry
when the spread is up to one order of magnitude larger than the average spread (i.e.\ it is less than $20$ bps).

\smallskip

Following the cascade calibration described in Section \ref{sec:Calibration}, we first calibrate the ATM vol and then the parameters $\eta$ and $k$ of the additive model, for the different value dates, with the choice of $\alpha=1/2$ (the results for other values of $\alpha$ are similar). 
In Figures \ref{fig:eta_all} and \ref{fig:k_all}, we show these two parameters for all value dates.

\smallskip

In Figure \ref{fig:eta_all}, we show the values of the parameter $\eta$ over the period of interest. 
This parameter moves in time, but slowly, and exhibits no significant oscillations.

\smallskip

Figure \ref{fig:k_all} shows the behaviour of the parameter $k$ for the different value dates during the analysed period. The level of $k$ remains close to one up to mid July\textcolor{black}{,} where it reached  --moving slowly over time-- a level close to $0.5$ where it remained up to the end of August\textcolor{black}{,} indicating relative calmer period of the Covid-19 pandemic during the summer time.
\begin{figure}[H]
    \centering
    \includegraphics[width=1\linewidth]{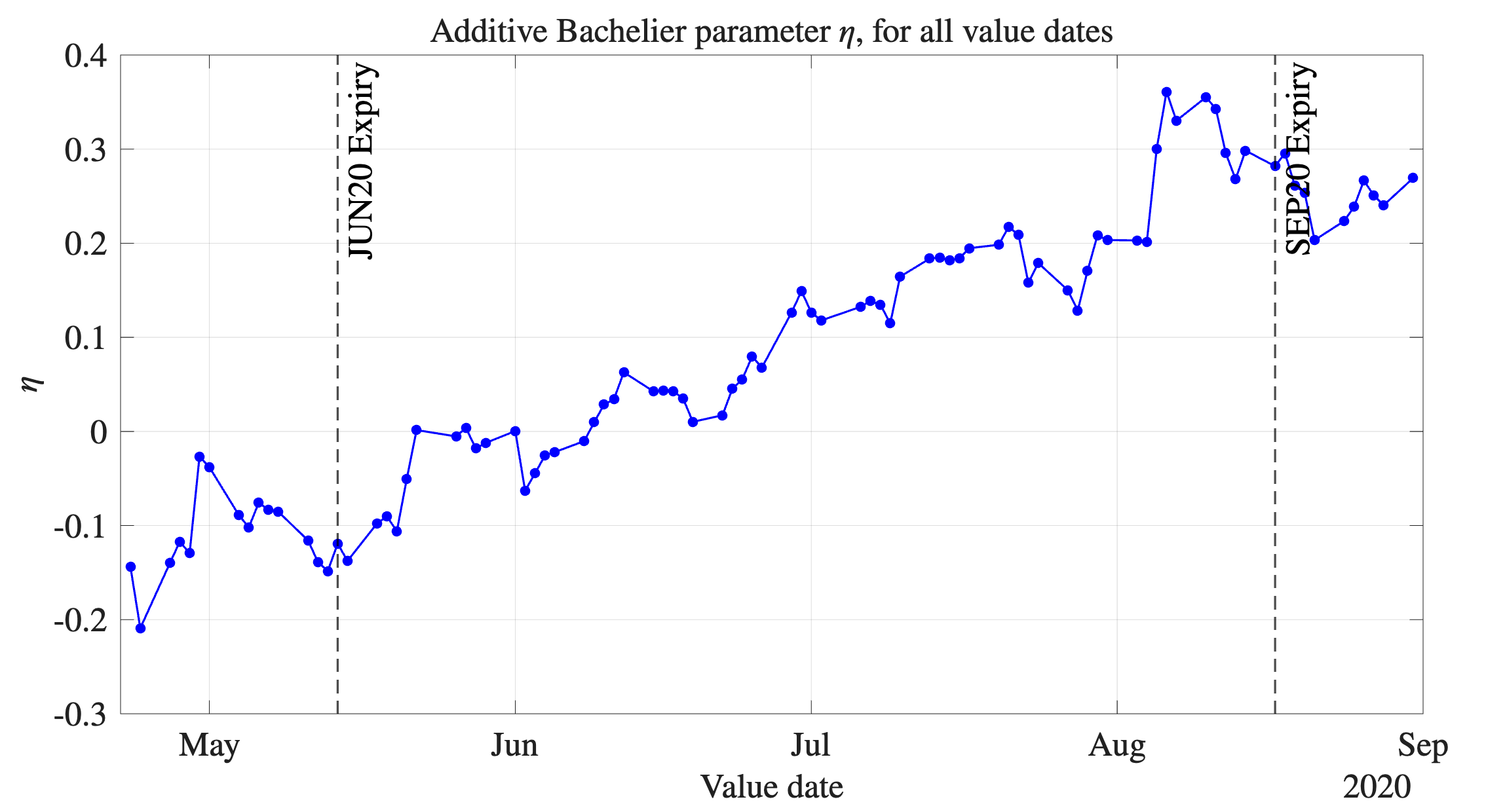}
    \caption{\small{Parameter $\eta$ of the additive Bachelier calibrated for the different value dates considered in the analysis. The vertical lines correspond to the expiries of the first two options. We observe that its value does \textcolor{black}{not} change significantly for the value dates, \textcolor{black}{and} it moves slowly over time.}}
    \label{fig:eta_all}
\end{figure}
\begin{figure}[H]
    \centering
    \includegraphics[width=1\linewidth]{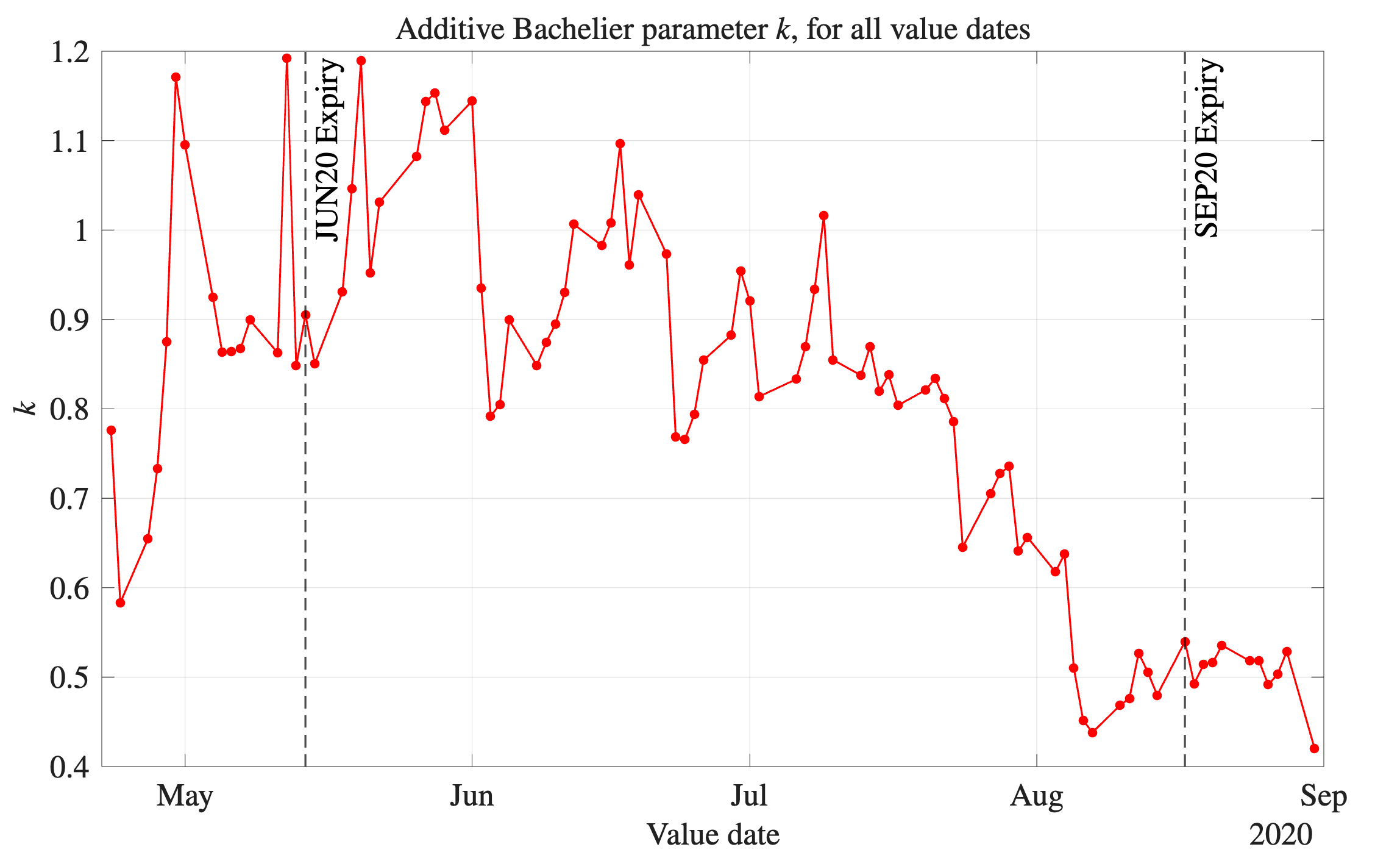}
    \caption{\small{Parameter $k$ of the additive Bachelier calibrated for the different value dates considered in the analysis. The vertical lines correspond to the expiries of the first two options. We observe that its value moves slowly over time.}}
    \label{fig:k_all}
\end{figure}

\subsubsection{Quality of calibration}
We have already commented on the exact calibration of forwards and ATM vols allowed by the cascade calibration. It is useful to underline the quality of the volatility surface description via only two parameters ($\eta$ and $k$) and their stability.
To quantitatively assess the stability of the calibrated parameters $\eta$ and $k$ over the different value dates $t_0$ present in the dataset, we consider the root mean square error (RMSE) between model and market prices, denoted respectively with $C^{mod}(x,t;\eta,k)$ and $C_{t_0}^{mkt}(x,t)$.
For each value date $t_0$ present in the dataset, we compute the error
\[
\displaystyle RMSE_{t_0}(\eta^*, k^*) :=\sqrt{\frac{1}{N_{t_0}} \sum_{i=1}^{N_{t_0}}\left[ C_{t_0}^{mkt} (x_i,t_i) - C^{mod}(x_i,t_i; \eta^*, k^*) \right]^2}\; ,
\]
where $N_{t_0}$ denotes the number of available European options on the value date $t_0$.\\
We consider the RMSE computing it on two sets of parameters $\{\eta^*,\,k^*\}$. First, we select the parameters equal to their optimal values $\{\eta_{t_0},\,k_{t_0}\}$
\[
\begin{cases}
\eta^*\leftarrow\eta_{t_0}&\\
k^*\leftarrow k_{t_0}&
\end{cases}\,\,,
\]
varying $t_0$ in the set of available value dates.\\
Then, we compute the RMSE on the pair of parameters obtained on the previous business day
\[
\begin{cases}
\eta^*\leftarrow\eta_{t_0-1}&\\
k^*\leftarrow k_{t_0-1}&
\end{cases}\,\,,
\]
and consider the difference of the two RMSE, indicated as $\Delta$RMSE in Figure \ref{fig:RMSE_robustness_day}.\\
As shown in Figure \ref{fig:RMSE_robustness_day}, the increase in error when using the previous value date's parameters is contained and always in the order of one cent or below. This fact suggests that the model parameters $\eta,\,k$ are stable over two consecutive value dates, i.e.\ their small deviations do not significantly worsen the model accuracy.

\begin{figure}[H]
    \centering
    \includegraphics[width=1\linewidth]{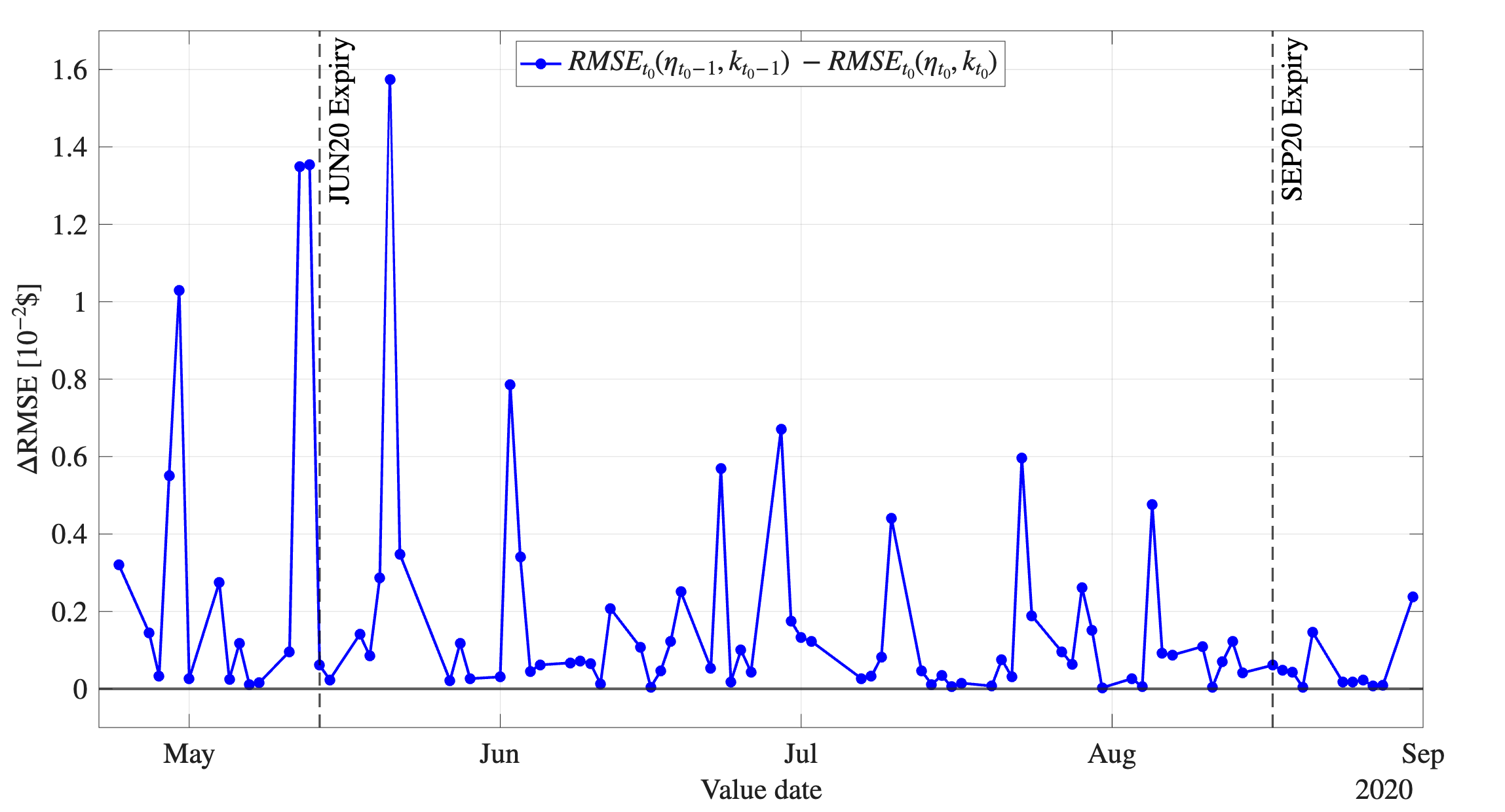}
    \caption{\small{For each value date $t_0$, \textcolor{black}{RMSE} between market prices and model prices using the parameters $\eta_{t_0},\,k_{t_0}$ calibrated for the value date $t_0$, $RMSE_{t_0}(\eta_{t_0},k_{t_0})$, and the RMSE between market prices and model prices obtained using the previous day parameters $\eta_{t_0-1},\,k_{t_0-1}$ $RMSE_{t_0}(\eta_{t_0-1},k_{t_0-1})$. The difference between the two errors is negligible for almost all dates (of the order of $1$ cent or below), confirming the stability over two consecutive value dates of the additive Bachelier model.}}
    \label{fig:RMSE_robustness_day}
\end{figure}


We further extend this analysis over a longer horizon, computing the average error over a rolling window of one week ($5$ business days). First, the values of the two parameters during one week are considered equal to the optimal ones
\[
\begin{cases}
\eta^*_s\leftarrow\eta_{s}&\\
k^*_s\leftarrow k_{s}&
\end{cases}\quad s=t_0-4,...,t_0\,\,.
\]
Then, the two parameters are kept constant over one week and equal to the ones obtained one week before
\[
\begin{cases}
\eta^*_s\leftarrow\eta_{t_0-5}&\\
k^*_s\leftarrow k_{t_0-5}&
\end{cases}\quad s=t_0-4,...,t_0\,\,.
\]
Specifically, we compare the average error using the parameters $\eta_s,\,k_s$, for each value date $s\in\left[t_0-4,\,t_0\right]$,
\[
\displaystyle RMSE_{t_0-4:t_0}(\eta_{t_0-4:t_0}, k_{t_0-4:t_0}) :=\sqrt{\frac{1}{\sum_{s=t_0-4}^{t_0}N_s}\sum_{s=t_0-4}^{t_0}\sum_{i=1}^{N_s}\left[ C_{s}^{mkt} (x_i,t_i) - C^{mod}(x_i,t_i; \eta_s, k_s) \right]^2} \; ,
\]
with the error obtained using fixed parameters $\eta_{t_0-5},\,k_{t_0-5}$ in the whole week,
\[
\displaystyle RMSE_{t_0-4:t_0}(\eta_{t_0-5}, k_{t_0-5}) :=\sqrt{\frac{1}{\sum_{s=t_0-4}^{t_0}N_s} \sum_{s=t_0-4}^{t_0}\sum_{i=1}^{N_s}\left[ C_{s}^{mkt} (x_i,t_i) - C^{mod}(x_i,t_i; \eta_{t_0-5}, k_{t_0-5}) \right]^2} \; ,
\]
calibrated on the value date the day before the beginning of the window. Then, we consider the difference $\Delta$RMSE between the two \textcolor{black}{RMSEs}. Results are presented in Figure \ref{fig:RMSE_robustness_week}, and continue to suggest that the additive Bachelier model maintains robust pricing performance even when the same parameters are used across 
all consecutive business days within a week. We observe that the $\Delta$RMSE across all value dates considered in the analysis remains consistently within a few cents of a dollar.\\
As already mentioned\textcolor{black}{,} commenting \textcolor{black}{on} Tables \ref{tab:put_call_couples} and \ref{tab:desc}, the WTI derivative market was relatively illiquid in the time window of interest.
We highlight that, despite the observed high volatility and low liquidity market, it is possible to achieve stable results over time in the calibration of the additive Bachelier parameters ($\eta$ and $k$). In the following Section, 
we discuss the calibration results achieved on a single value date. 

\smallskip

This fact suggests that the additive Bachelier model exhibits not only parsimony but also temporal stability, a crucial property in a volatile and relatively illiquid market like the WTI oil option market during the first period of the Covid pandemic.
\begin{figure}[H]
    \centering
    \includegraphics[width=1\linewidth]{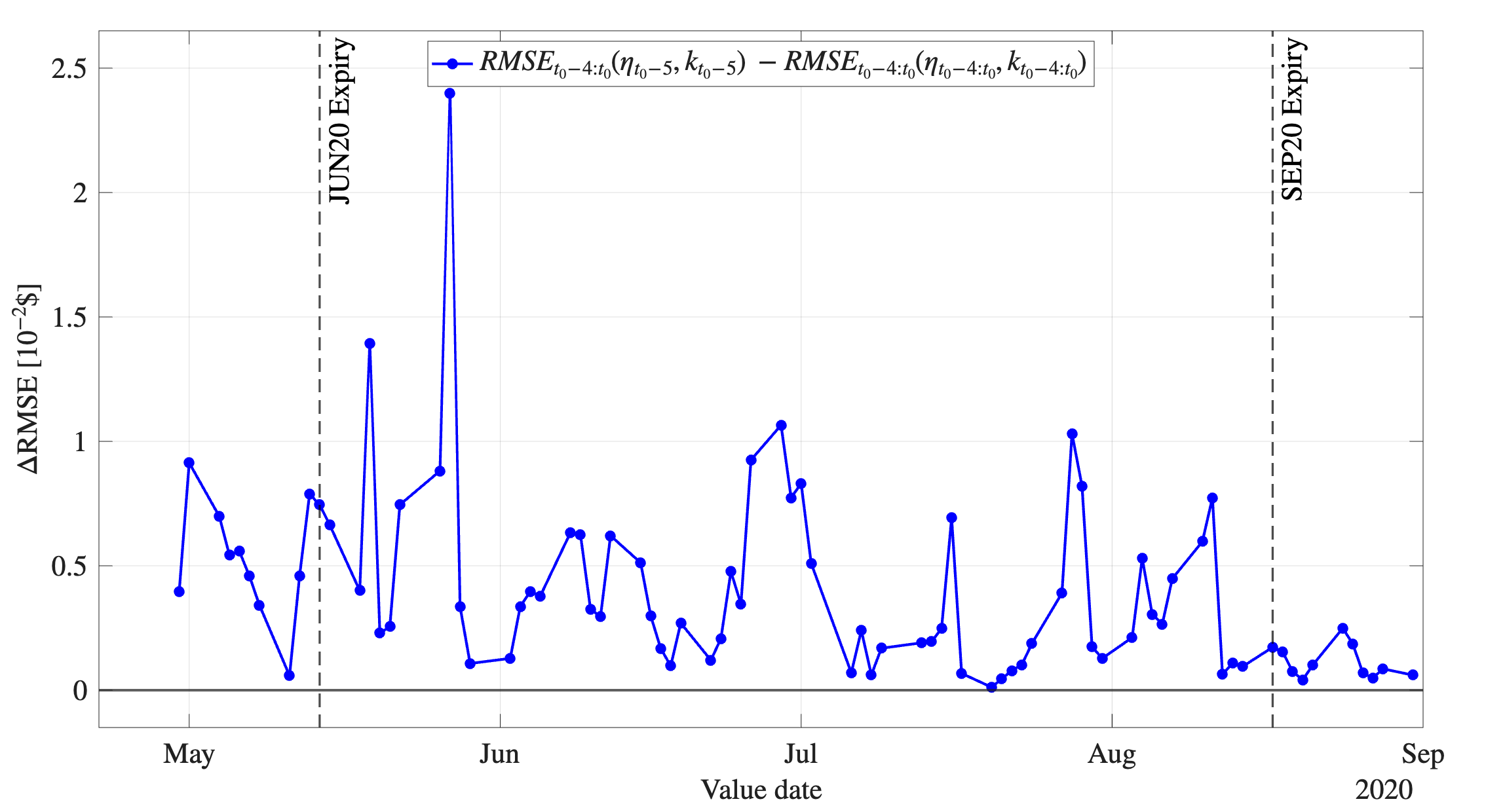}
    \caption{\small{Difference of RMSE over a $5$-day window using fixed versus \textcolor{black}{optimal} parameters. We compare the $5$-day RMSE between market prices and model prices using the parameters $\eta_s,\,k_s$ for the each value date $s\in\left[t_0-4,\,t_0\right]$, $RMSE_{t_0-4:t_0}(\eta_{t_0-4:t_0}, k_{t_0-4:t_0})$, against the $RMSE_{t_0-4:t_0}(\eta_{t_0-5}, k_{t_0-5})$ obtained using the fixed parameters $\eta_{t_0-5},\,k_{t_0-5}$, calibrated before the beginning of the window. We observe a contained increase in the calibration error of about $1$ cent or below, with the only exception of the $27^{\mathrm{th}}$ of May 2020, when it is less than $2.5$ \textcolor{black}{cents}. This fact highlights the robustness of the additive Bachelier model’s parameters.}}
    \label{fig:RMSE_robustness_week}
\end{figure}

\subsection{Detailed results for one date: the $\mathbf{29}^\mathbf{th}$ of April 2020}\label{sec:one_date}

In this Section, we show the results with value date on the $29^\mathrm{th}$ of April 2020 (one week after the \textcolor{black}{announcement} from the CME, in one of the most volatile periods),\footnote{Results do not appear significantly different on all other value dates and they are available upon request.} compared with the results of other three models considered as benchmarks:
\begin{itemize}
	\item the L\'evy Bachelier model (hereinafter ``L\'evy model"), defined in (\ref{eq:ChFunLevy}), Appendix \ref{app:Bachelier_formula},
	\item a L\'evy Bachelier model that includes options with one ttm at a time (named ``slice-by-slice"),
	\item an additive logistic model \citep[][]{CarrTorricelli2021}.
\end{itemize}

The first benchmark is the simplest generalization of Bachelier model that allows for a volatility surface, 
while the second is the same model but calibrated considering separately the options with a given ttm: the latter utilizes a different set of parameters for each maturity and it is not in general a coherent model for the whole volatility surface but --as we'll show in this Section-- it reproduces quite accurately the market volatility surface.  The last one is a recent model with an additive process applied to a real-valued underlying.

\smallskip

In Figure \ref{fig:sintRates}, we compare the discount-rates (from option prices) with the ``risk-free" rates from the Effective Federal Funds Rate OIS curve. We see that the spread on the first expiry is larger than 20 basis points, while the spread corresponding to the other expiries is much smaller.

\smallskip

As already mentioned, since the spread on the first date is larger than $20$ \textcolor{black}{bps}, we do not consider the first expiry in the next steps of the calibration.

\begin{figure}[H]
    \centering
    \includegraphics[width=1\linewidth]{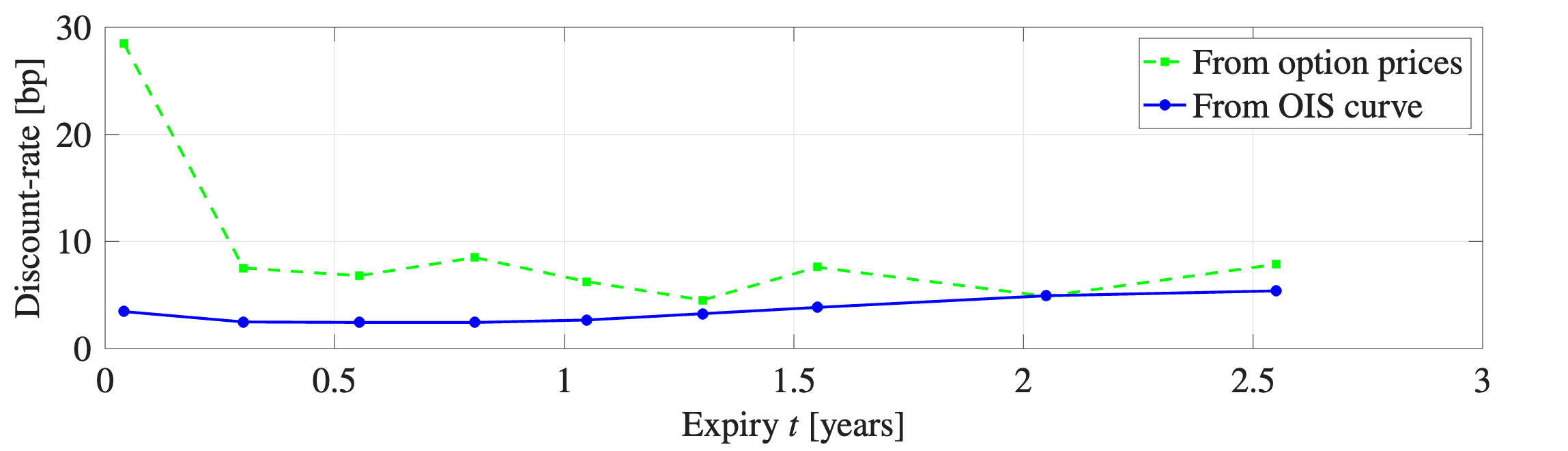}
    \caption{\small{Discount-rates from option prices (in green, squares and dashed line), and the ``risk-free" rates from the OIS curve (in blue, dots and continuous line) observed in the market on the $29^\mathrm{th}$ of April 2020.}}
    \label{fig:sintRates}
\end{figure}

After replicating discount-rates and ATM forward, we consider the ATM volatility.

\smallskip

In Figure \ref{fig:ATM_vol_An2}, we plot the ATM volatility from market implied volatility (red, squares and continuous line) with the ATM volatility from the calibration of the L\'evy model (black dashed line), the additive logistic (cyan continuous line), the slice-by-slice (green dot-dashed lines) and the additive Bachelier model (blue dots and continuous line). The slice-by-slice and the additive Bachelier model reproduce exactly the ATM volatility term structure, while both the L\'evy model and the additive logistic are not able to replicate the ATM volatility term structure.

\begin{figure}[H]
    \centering
    \includegraphics[width=1\linewidth]{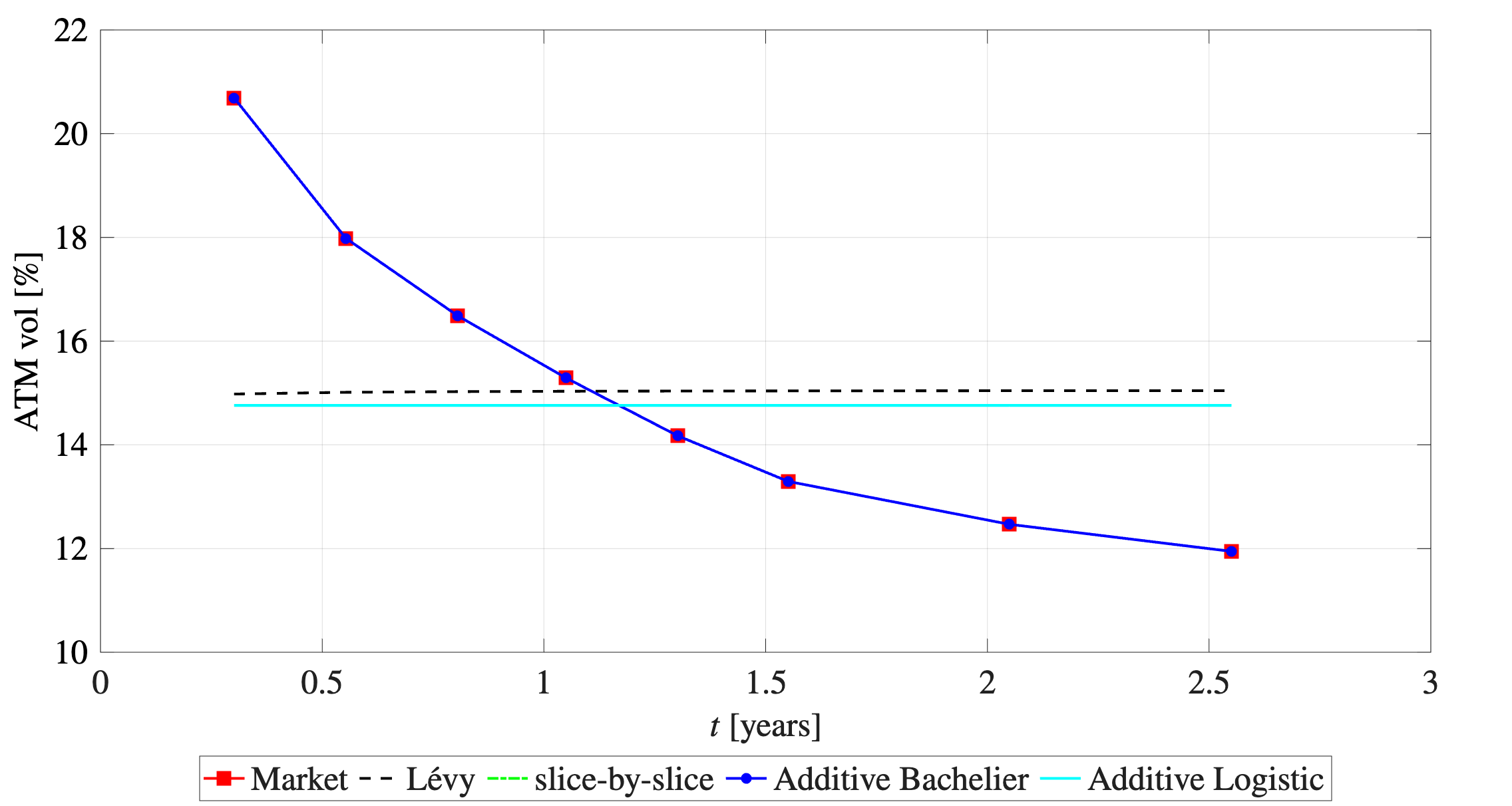}
    \caption{\small{ATM volatility comparison between market data (red, squares and continuous line), the additive model (in blue, dots and continuous line), the slice-by-slice (green, dots dashed line), the L\'evy model (black, dashed line) and the additive logistic (cyan, continuous line), on the $29^\mathrm{th}$ of April 2020, for $\alpha=1/2$.
While the first three are indistinguishable in the Figure\textcolor{black}{,} with both the slice-by-slice and the additive Bachelier replicating exactly \textcolor{black}{the} market ATM vol, both
the L\'evy model and the additive logistic are not able to reproduce the ATM vol.}}
    \label{fig:ATM_vol_An2}
\end{figure}

Then, we can calibrate the implied volatility surface for the models discussed. In Figure \ref{fig:surf_a_fette_An2}, we plot the volatility smile for all option maturities (without considering the first one), with value date the $29^\mathrm{th}$ of April 2020. We consider all options with moneyness $x$ in the range $[-30\,\$,30\,\$]$. We show the results of the L\'evy Bachelier (black dashed lines), the slice-by-slice (green dot-dashed lines), the additive logistic (cyan continuous lines), and the additive Bachelier (blue dots), considering $\alpha=1/2$. We observe that results from the L\'evy model and the additive logistic model cannot be considered adequate, since they fail to replicate the volatility smile in all eight expiries. The calibration provided by slice-by-slice replicates market data accurately. Finally, the additive Bachelier provides a quite good calibration of the implied volatility surface over all considered expiries despite its parsimony.

\begin{figure}[H]
    \centering
    \includegraphics[width=1\linewidth]{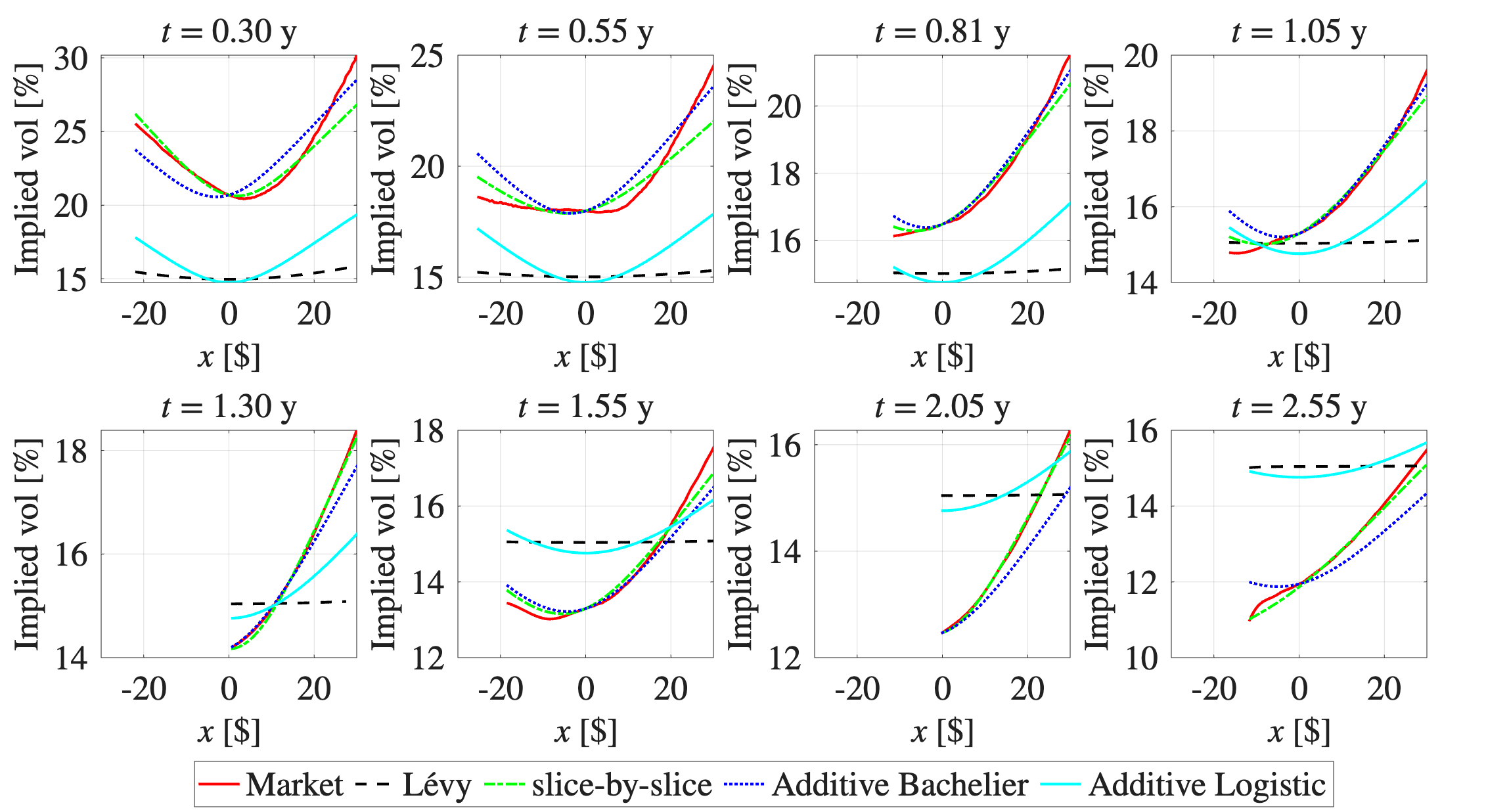}
    \caption{\small{Volatility surface. We plot the volatility smile for each option maturity present in the dataset after the first one, 
    for the value date $29^\mathrm{th}$ of April 2020. We consider all options with moneyness $x$ in the interval $[-30\,\$,30\,\$]$, for some maturities this moneyness interval is even larger than the observed market prices. Market implied vol (red, continuous lines) is plotted together with the three models discussed in the text: the L\'evy Bachelier (black, dashed lines), the L\'evy that includes options with one ttm at time (slice-by-slice, green, dot-dashed lines), the additive logistic model (cyan, continuous lines) and the additive Bachelier (blue, dots). While results for the L\'evy and additive logistic model cannot be considered adequate, the calibration provided by slice-by-slice replicates accurately market data, and the additive Bachelier provides values quite close to the latter. We consider $\alpha=1/2$.}}
    \label{fig:surf_a_fette_An2}
\end{figure}

We compare the results for the four models observing the \textcolor{black}{RMSE} of model prices wrt market prices separately for each maturity $t$ (with moneyness in the interval $[-30\,\$,30\,\$]$). We plot in Figure \ref{fig:ATM_RMSE_An2} the RMSE from each model for all the expiries (except for the first one) present in the dataset.

\smallskip

The best fit is achieved by the slice-by-slice description. However, it is not a coherent model for the whole volatility surface, as previously stated.

\smallskip

We observe that the error from the L\'evy model and the additive logistic are almost one order of magnitude larger than the RMSE of the slice-by-slice, while the error of the additive Bachelier (in blue, dots and continuous line) is quite close to the one of the slice-by-slice for expiries up to one year, becoming larger for longer expiries.

\begin{figure}[H]
    \centering
    \includegraphics[width=1\linewidth]{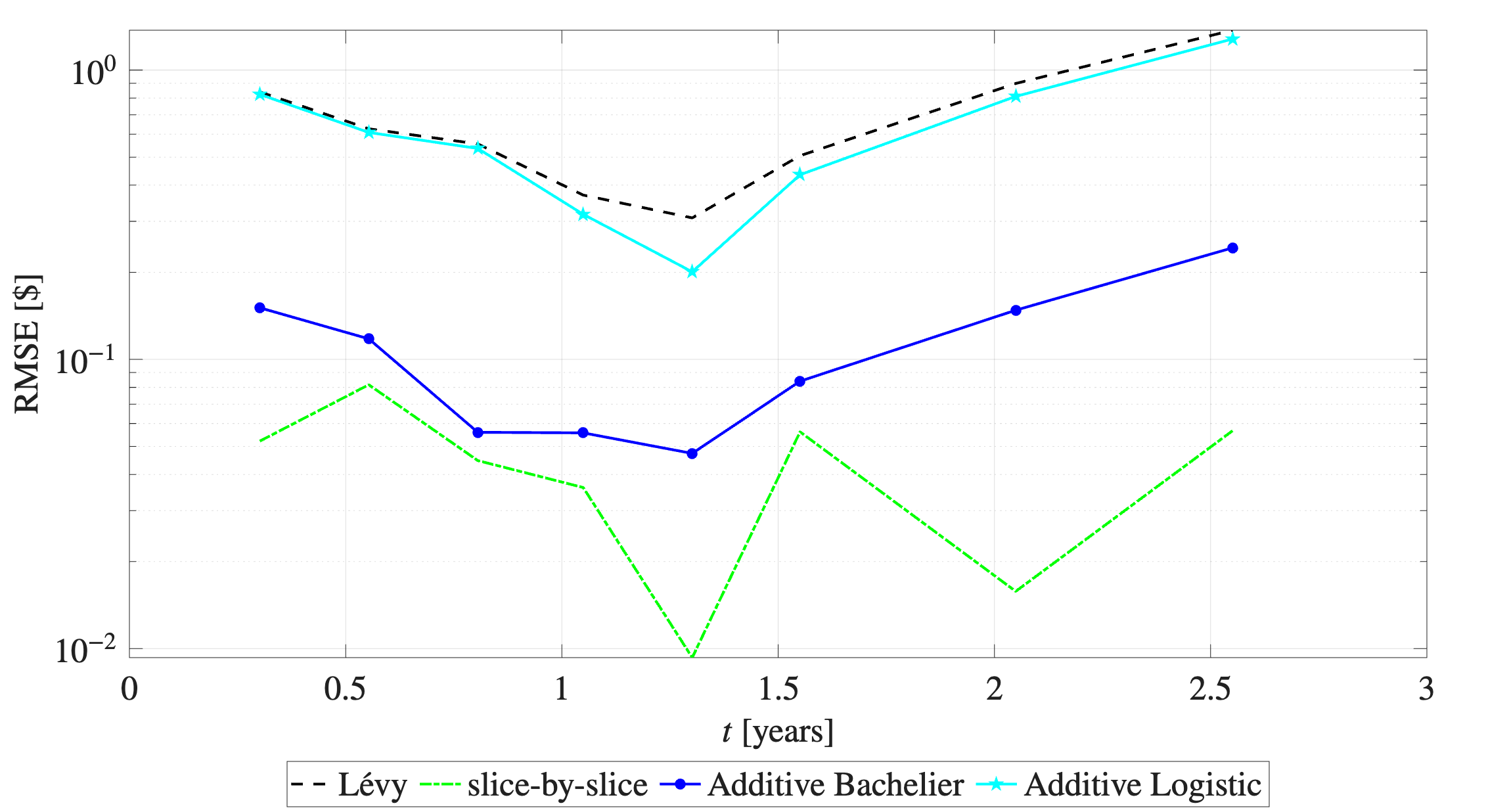}
    \caption{\small{Comparison of the RMSE obtained when calibrating the L\'evy model (in black, dashed line), the slice-by-slice (green, dots dashed line), the additive logistic model (cyan, stars and continuous line) and the additive Bachelier (blue, dots and continuous line) in the $29^\mathrm{th}$ of April 2020, for $\alpha=1/2$. The L\'evy model and the additive logistic present an RMSE almost one order of magnitude larger than the slice-by-slice, while RMSE for additive Bachelier is rather close (especially for options with maturity up to $1$ year and a half) to the slice-by-slice.}}
    \label{fig:ATM_RMSE_An2}
\end{figure}

\smallskip

Figure \ref{fig:parameters_An2} shows the calibrated parameters $\eta$ and $k$ of the additive Bachelier (blue, continuous, constant lines). The green squares (with dashed lines) correspond to the parameters of the slice-by-slice, for all the expiries. The close alignment between the two kinds of parameters suggests that the additive Bachelier parameters exhibit consistency across maturities.

\begin{figure}[H]
    \centering
    \includegraphics[width=1\linewidth]{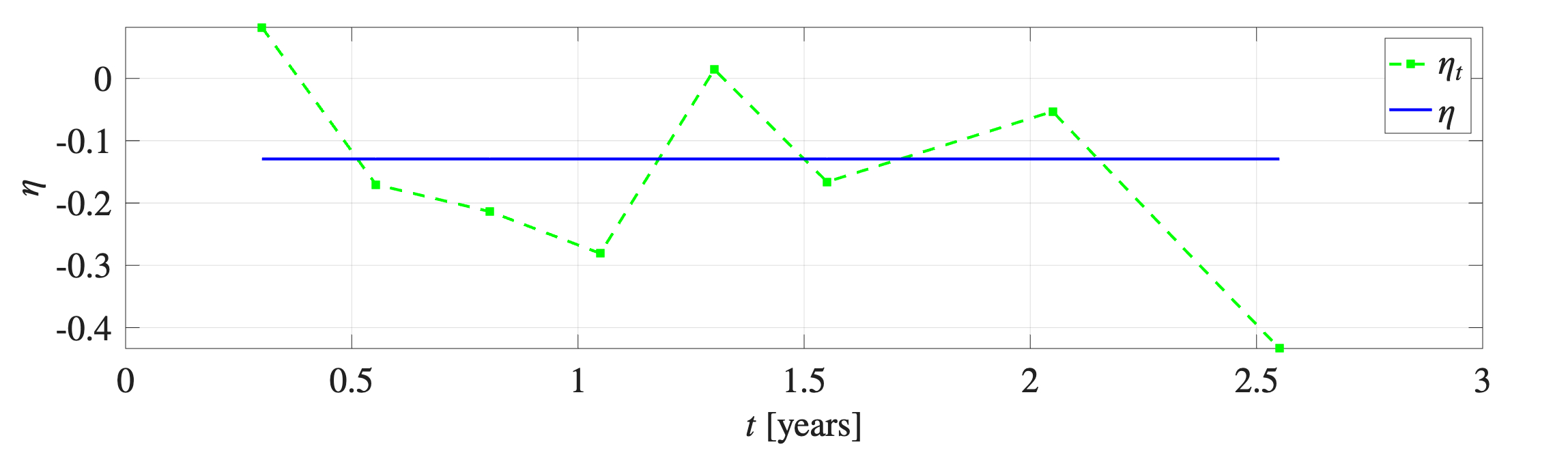}
    \includegraphics[width=1\linewidth]{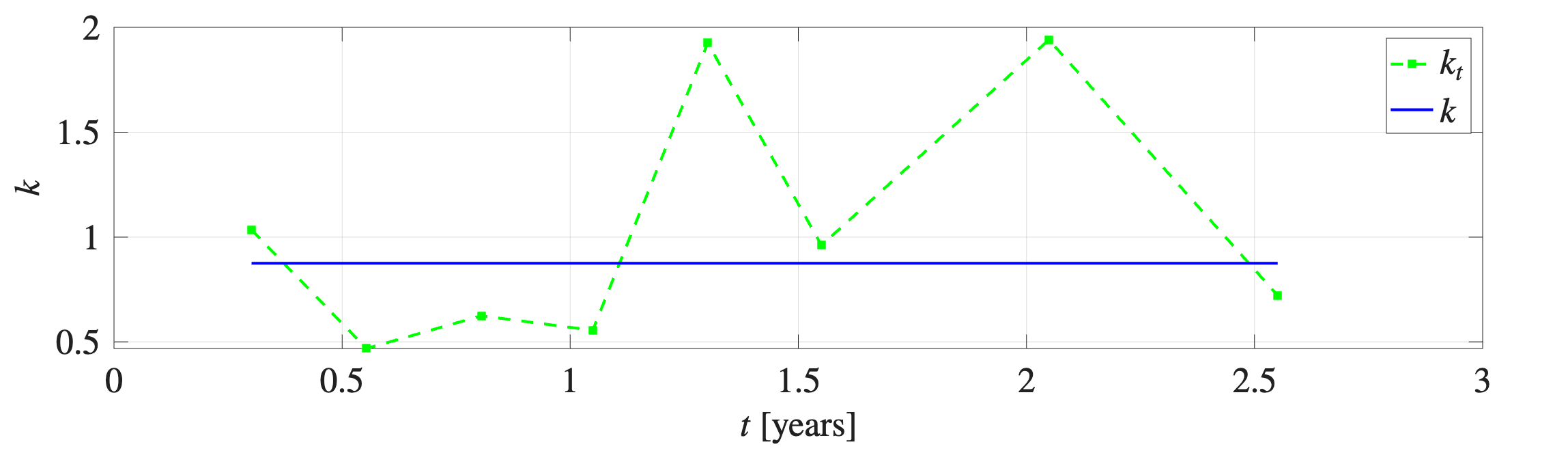}
    \caption{\small{Parameters $\eta,k$ of the additive Bachelier (blue continuous lines), compared with the parameters $\eta_t,\,k_t$ of the slice-by-slice. They appear quite similar for all options after the first expiry traded in the market. The top panel shows the parameter  $\eta$, and the bottom panel shows the parameter $k$.}}
    \label{fig:parameters_An2}
\end{figure}

\section{Conclusions}
\label{sec:Conclusions}
This paper has introduced the additive Bachelier, an extension of the Bachelier model, to describe the 
WTI oil options.

\smallskip

The proposed model allows for negative prices, and it is very simple from an analytical point of view, 
allowing for explicit expression of the implied vol close to ATM and for very large moneyness. 
Key features of the model include:
\begin{itemize}
    \item A closed-form solution for European options pricing, the Bachelier-Lewis formula.
    \item A cascade calibration approach that prioritizes the most liquid instruments, ensuring consistency in the pricing of ATM forward prices, ATM volatilities, and the implied volatility surface.
\item The ability to synthesize the entire volatility surface using just two parameters, each of them with a clear financial interpretation, facilitating practical application.
\item The robustness of the model's parameters across different value dates, that ensures consistent model accuracy over time.
\end{itemize}

The calibration results (discussed in Section \ref{sec:Results}) highlight the robustness of the proposed model over a very volatile period as the first period of the Covid-19 pandemic, ensuring a stable calibration for different value dates in the whole period from April 2020 to August 2020, and reproducing quite closely key market features.

\smallskip

\textcolor{black}{While this study focuses on the WTI oil market during the Covid period, the proposed modelling framework admits natural extensions to other asset classes and market conditions, which are left as future developments for this research. First, its ability to handle negative prices makes it suitable for other commodities such as natural gas \citep[e.g.,][]{CMEGroup} and electricity \citep[e.g.,][]{ice2024bachelier}, which can experience extreme downside volatility or near-zero prices. Second, given its roots in the Bachelier framework, the model is particularly well-suited for interest rate derivatives (e.g., swaptions, caps/floors) in low or negative interest rate environments, offering a flexible volatility smile description beyond the classic normal model. Furthermore, the parsimony of the model --relying on only a few parameters-- suggests potential scalability to multi-curve or multi-commodity settings, where calibration speed is crucial. Finally, the additive process structure naturally accommodates time-inhomogeneity, paving the way for dynamic extensions that can capture regime shifts, a relevant feature for energy markets undergoing structural changes.}

\section*{Acknowledgements}
We thank all participants to the EFI10 Conference, the Advances in Mathematical Finance Conference in Freiburg,  the 12th General AMaMeF Conference and the XLIX Annual Conference of the AMASES, in particular Giovanni Amici, Laura Ballotta, Òscar Burés, Ernst Eberlein, Friedrich Hubalek and Ioannis Kyriakou for comments and discussions.
\section*{Appendix \ref{app:Bachelier_formula}: Bachelier formula, L\'evy Bachelier and its generalization}
\customlabel{app:Bachelier_formula}{A}
\subsection*{Bachelier formula basic properties}

We briefly recall some basic properties of the normalized Bachelier formula $ c_b (y, \sigma)$ defined in \eqref{eq:BasicFormula}. It can be obtained as
\[
c_b (y, \sigma) = \E[ - \sigma g -y]^+
\]
with $g$ a standard normal rv.

It is a positive homogeneous function (of degree one)
\be
 c_b(\lambda \, y, \lambda \, \sigma)  =  \lambda \, c_b(y, \sigma) \quad  \forall \lambda>0 \;\; .
\label{eq:hom}
\en

The time-value, defined as the difference between price and payoff, is
\[
TV(y,\sigma) =- |y| \, \Phi \left( - \frac{|y|}{\sigma} \right) + \sigma \, \varphi \left( - \frac{|y|}{\sigma} \right)\,\,,
\]
it is symmetric \textcolor{black}{wrt} the \textcolor{black}{moneyness degree}.

\subsection*{Bachelier Greeks}

\be
\label{eq:Vega}
\begin{array}{lp{2cm}rcccl}
{\rm Delta} & & \Delta_b & := & \displaystyle \frac{\partial}{\partial (-y)} c_b (y, \sigma) & = & \displaystyle  \Phi \left( - \frac{y}{\sigma} \right) \;\; , \\[3mm]
{\rm Vega} & & {\cal V}_b & := & \displaystyle  \frac{\partial}{\partial \sigma} c_b (y, \sigma) & = & \displaystyle  \varphi \left( - \frac{y}{\sigma} \right) \;\; , \\[3mm]
{\rm Gamma} & & \Gamma_b & := & \displaystyle  \frac{\partial^2}{\partial (-y)^2} c_b (y, \sigma)  & = & \displaystyle \frac{1}{\sigma} 
\varphi \left( - \frac{y}{\sigma} \right)  \;\; , \\[3mm]
{\rm Vanna} & & vanna_b & := & \displaystyle   \frac{\partial}{\partial (-y)}
\frac{\partial}{\partial \sigma} 
c_b (y, \sigma)   & = & \displaystyle \frac{y}{\sigma^2}
\varphi \left( -  \frac{y}{\sigma}\right)   \;\; .
\end{array}
\en

\subsection*{L\'evy Bachelier model}
The Bachelier model is not able to describe any volatility surface.
The simplest generalization is a Bachelier L\'evy model defined through its normal tempered stable characteristic function
\begin{equation}
\label{eq:ChFunLevy}
	\ln \phi_t(u):= \ln \mathbb{E}_0\left[e^{ i \, u \, f_t}\right]=
	t\left(\psi\left(i \, u \, \hat{\eta} \,  \hat{\sigma}
+\frac{u^2 }{2} \, \hat{\sigma}^2 ;  \hat{k}, \; \alpha \right) \, + \,
 i\, u \, \hat{\eta}\,\hat{\sigma} \right)\;\;,
\end{equation}
with $\hat{\eta} \in \mathbb{R}$ and $\hat{\sigma}, \hat{k} \in \mathbb{R}^+$\textcolor{black}{, and $\psi(\bullet;k,\alpha)$ is defined in \eqref{eq:laplaceG}.}
This model is the generalization to the linear (Bachelier) case of the normal tempered stable 
 exponential L\'evy model that can be found in excellent textbooks \citep[see, e.g.,][]{Cont}.
 
This model, unfortunately, is not able to reproduce the observed volatility surface \textcolor{black}{adequately}, as also shown in Section \ref{sec:Results}.
However, a L\'evy model that is calibrated on one single maturity at a time (\textcolor{black}{also} called slice-by-slice) is able to calibrate in an excellent way
even the oil derivative market in the most turbulent period  that has been observed so far: 
 the first months of the Covid pandemic.
Clearly, this model has a set of parameters that vary at each maturity and it is not necessarily well posed.

For this reason, we would like to consider the natural inhomogeneous extension of the L\'evy process case \eqref{eq:ChFunLevy}, where the parameters are time dependent, as described in the following Lemma.
The proof of the following Lemma is analogous to \textcolor{black}{that} of Theorem 2.1 in \citet[][Thm.2.1]{ATS}, valid for exponential additive processes\textcolor{black}{, and it is provided in Appendix \ref{app:Proof_1}.}
{\setcounter{lemma}{0}
\renewcommand{\thelemma}{A.\arabic{lemma}}
\begin{lemma}{\rm (A class of additive processes $\{ f_t \}$)} 
\label{lem:f_AdditiveGen}\\
Given the continuous deterministic functions of time $\sigma_t$ strictly positive and such that $\sigma_t^2t$ goes to zero for $t\rightarrow0^+$, $\eta_t \in \R$ and $k_t \in \R^+$, and $\alpha\in [0,1)$,
the class of processes $\{ f_t \}_{t\ge 0}$ with characteristic function exponent
\begin{equation}
\label{eq:ChFunGeneral}
	\ln \phi_t(u):= \ln \mathbb{E}\left[e^{ i \, u \, f_t}\right]=\psi\left(i \, u \, \eta_t \, {\sigma }_t  \, \sqrt{t}   
+\frac{u^2 }{2} \, {\sigma }_t^2 \, t ; \; k_t, \; \alpha \right) \, + \,
 i\, u \, \eta_t \, {\sigma }_t \, \sqrt{t}     \;\;,
\end{equation}
\textcolor{black}{where the function $\psi$ is defined as in \eqref{eq:laplaceG},}
is additive when the following conditions hold true: 	
\begin{enumerate}
\item The functions  
\[
\begin{array}{lcl}
p^\pm_t  &:=& \displaystyle \frac{1}{\sigma_t \sqrt{t}} \left\{ \displaystyle \lr \eta_t + \sqrt{\eta_t^2+2 \, \frac{1-\alpha}{k_t} }  \right\}\\[3mm]
p^*_t      &:=& \displaystyle -\frac{\sigma_t\sqrt{t}}{k_t^{(1-\alpha)/\alpha}}\sqrt{\eta_t^2+ 2\,\frac{1-\alpha}{ k_t}}\;\;
\end{array}
\]
are non-increasing;
\item \(\displaystyle \sqrt{t} \,\sigma_t\, \eta_t \) and \(\displaystyle \sqrt{t}\,\sigma_t \, \eta_t \, {k_t^{(\alpha -1)/\alpha}} \) go to zero as $t$ goes to zero.
\end{enumerate}
\end{lemma}}

\smallskip

The L\'evy Bachelier model is the simplest case of the additive process \eqref{eq:ChFunGeneral}, that can be obtained by choosing 
\[
\sigma_t= \hat{\sigma}  ,  \, \eta_t = \hat{\eta} \sqrt{t}  , \, k_t = \frac{ \hat{k} }{t} \;\; .
\]

In this case, condition $1$ of Lemma \ref{lem:f_AdditiveGen} is satisfied because $p^\pm_t$ are constant in time, while $p^*_t$ is decreasing, and
condition $2$ holds; thus, the process is additive. 
Moreover, the process \textcolor{black}{has stationary increments} because the log-characteristic function \eqref{eq:ChFunLevy} is linear in $t$ \textcolor{black}{\citep[see, e.g.,][Def.3.1, Sec.3.1]{Cont}}.

\smallskip

The  additive Bachelier presented in this paper is another particular case of the additive process \eqref{eq:ChFunGeneral} that, even being very parsimonious, i) presents the interesting analytical properties described in this study and ii) is able to obtain calibration results quite close to the more general 
(but not necessarily well posed) slice-by-slice\textcolor{black}{,} that has time dependent parameters that vary at each maturity. In the additive Bachelier, $\sigma_t$ is s.t.\ i) $\sigma_t^2 t$ is monotone (property naturally satisfied by the ATM vol in real markets), ii) $\sigma_t^2 t$ goes to zero as $t\to0^+$, and
\[
\eta_t = \eta, \, k_t = k \;\; .
\]
This observation allows us to prove Proposition \ref{th:f_Additive}, as shown in Appendix \ref{app:Proof_1}.

\smallskip

Finally, let us observe that
the Bachelier model corresponds to the limiting case of the additive Bachelier for $k \to 0^+$.

\section*{Appendix \ref{app:Proof_1}: Proofs}
\customlabel{app:Proof_1}{B}
\begin{proof} {\bf Lemma \ref{lem:f_AdditiveGen}} 


At any given time $t>0$, we have the expression for the generating triplet for a process with characteristic function exponent (\ref{eq:ChFunGeneral}) \citep[see, e.g.,][eq.(4.24), p.130]{Cont}
\begin{equation*}
\begin{cases}
A_t&=0\\
\gamma_t&=\int_{|x|<1}x\,\nu_t(\mathrm{d}x)+\eta_t\sigma_t\sqrt{t}\\
\nu_t(x)&=\frac{C(\alpha,k_t,\sigma_t,\eta_t)}{|x|^{\alpha+1/2}}\,e^{-\frac{\eta_t\,x}{\sigma_t\sqrt{t}}}\,K_{\alpha+1/2}\left(|x|\frac{\sqrt{\eta_t^2+2\frac{1-\alpha}{k_t}}}{\sigma_t\sqrt{t}}\right)
\end{cases}
\end{equation*}
where
\begin{equation*}
C(\alpha,k_t,\sigma_t,\eta_t)=\frac{2}{\Gamma(1-\alpha)\sqrt{2\pi}\sigma_t\sqrt{t}}\left(\frac{1-\alpha}{k_t}\right)^{1-\alpha}\left(\sqrt{\eta_t^2\sigma_t^2t+2\sigma_t^2t\frac{1-\alpha}{k_t}}\right)^{\alpha+1/2}\,\, ,
\end{equation*}
and
\begin{equation*}
K_\theta(z)=\frac{e^{-z}}{\Gamma(\theta+1/2)}\sqrt{\frac{\pi}{2z}}\int_0^{+\infty}e^{-s}s^{\theta-1/2}\left(\frac{s}{2z}+1\right)^{\theta-1/2}\mathrm{d}s
\end{equation*}
is the modified Bessel function of the second kind \citep[see, e.g.,][Ch.9, p.376]{abramowitz1948handbook}. For $t=0$, we set $A_0=0$, $\gamma_0=0$ and $\nu_0=0$.

\bigskip

First, we prove that $\nu_t(x)$ is a non-decreasing function wrt $t$. The expression for $\nu_t(x)$ is the following 

\begin{align*}
\nu_t(x)=&\frac{(1-\alpha)^{1-\alpha}}{\Gamma(1-\alpha)\Gamma(1+\alpha)|x|^{\alpha+1}}\exp\left\{-\frac{1}{\sigma_t\sqrt{t}}\left(x\eta_t+|x|\sqrt{\eta_t^2+2\frac{1-\alpha}{k_t}}\right)\right\}\\
&\int_0^{+\infty}e^{-s}s^{\alpha}\left[\frac{\sigma_t^2t}{k_t^{(1-\alpha)/\alpha}}\left(\frac{s}{2|x|}+\frac{1}{\sigma_t\sqrt{t}}\sqrt{\eta_t^2+2\frac{1-\alpha}{k_t}}\right)\right]^{\alpha}\mathrm{d}s\,\,.
\end{align*}

\smallskip

The function $\nu_t(x)$ is non-decreasing wrt $t$ $\forall x$ if the following two terms are not decreasing
\begin{align}
&\label{eq:terms_exp}\exp\left\{-\frac{1}{\sigma_t\sqrt{t}}\left(x\eta_t+|x|\sqrt{\eta_t^2+2\frac{1-\alpha}{k_t}}\right)\right\}\,\,,\\
&\label{eq:terms_int}\frac{\sigma_t^2t}{k_t^{(1-\alpha)/\alpha}}\left(\frac{s}{2|x|}+\frac{1}{\sigma_t\sqrt{t}}\sqrt{\eta_t^2+2\frac{1-\alpha}{k_t}}\right)\,\, .
\end{align}

The condition on $p^\pm_t$ implies that (\ref{eq:terms_exp}) is not decreasing wrt $t$ $\forall x\in\R$. The condition on $p^\pm_t$ implies also that
\begin{equation*}
\frac{1}{\sigma_t\sqrt{t}}\sqrt{\eta_t^2+2\frac{1-\alpha}{k_t}}
\end{equation*}
is non-increasing wrt $t$, thus the factor $\sigma_t^2\,t\,k_t^{(\alpha-1)/\alpha}$ is not decreasing and that also (\ref{eq:terms_int}) is not decreasing.

\bigskip

To prove that
\begin{equation}
\label{eq:lim_nu}
\lim_{t\rightarrow0^+}\nu_t(x)=0\quad\forall x\neq0\,\,,
\end{equation}
we discuss three cases (whose union covers all possible cases)
\begin{align*}
\mathrm{a)}\quad&\lim_{t\rightarrow0^+}\eta_t^2k_t=0\,\, ,\\
\mathrm{b)}\quad&\lim_{t\rightarrow0^+}k_t=0\,\, ,\\
\mathrm{c)}\quad&\lim_{t\rightarrow0^+}\eta_t^2k_t\neq0\quad\mathrm{and}\quad\lim_{t\rightarrow0^+}k_t\neq0\,\, ,
\end{align*}
we show that (\ref{eq:lim_nu}) holds on all the three cases. \\
In cases a) and b), it's straightforward to prove that the term in (\ref{eq:terms_exp}) goes to zero faster than (\ref{eq:terms_int}) diverges. In the last case c), thanks to hypothesis 2, the term in (\ref{eq:terms_int}) goes to zero, independently on the other term, since the term (\ref{eq:terms_exp}) is an exponential with negative exponent $\forall\,x\in\R\setminus\{0\}$ and $\forall\,t>0$, thus this term cannot diverge.
\smallskip
Thus, the following holds
\begin{equation*}
\lim_{t\rightarrow0^+}\nu_t(x)=0\quad\forall x\neq0\,\,.
\end{equation*}
Moreover,
\begin{equation*}
\lim_{t\rightarrow0^+}\gamma_t=0\,\, \textcolor{black}{,}
\end{equation*}
due to dominated convergence Theorem and hypothesis 2.

\bigskip

We can now check whether the triplet satisfies the conditions in \citet[Thm.9.8, p.52]{Sato}.
\begin{enumerate}
\item The triplet has no diffusion term.
\item $\nu_t$ is non-decreasing wrt $t$.
\item $\nu_t(B)$ and $\gamma_t$ are continuous, where $B\in \mathbb{B}(\R^+)$ and $B\subset\{x:|x|>\epsilon>0\}$. The continuity for $t>0$ is obvious, since it is a natural consequence of the composition of continuous functions. We have to prove that the limits of $\nu_t(B)$ and $\gamma_t$ are $0$. We have already proven that $\nu_t(x)$ is non decreasing in $t$ and that $\lim_{t\to0^+} \nu_t(x) = 0,\,\forall x\neq 0$. The convergence of $\nu_t(B)$ to $0$ is due to the dominated convergence theorem.
\end{enumerate}
\end{proof}

Before proving that the proposed process is additive and martingale, 
let us mention an important property of the model, stated in the following Lemma.

\smallskip
\setcounter{lemma}{0}
{
\renewcommand{\thelemma}{B.\arabic{lemma}}
\begin{lemma} {\rm (Equivalence in law of the process $\{ f_t \}$)}\label{lem:f_equivalence}\\
The process $ \left\{f_t \right\}_{t\geq 0}$ with characteristic function \eqref{eq:ChFun} --for any given $t$-- is equivalent in law to
\be
\label{eq:EquivInLaw}
{\sigma }_t \, \sqrt{ t} \, z\,\,,
\en
with
\begin{equation}
\label{eq:z_rv}
z:= \eta \,  (1 - G) -    \sqrt{G} \, g \; ,    
\end{equation}
where $g$ is a standard normal rv and $G$ a positive rv with Laplace exponent \eqref{eq:laplaceG}. The characteristic function of $z$ is
\begin{equation*}
\phi^{(z)}(u)=\phi_t\left(\frac{u}{\sigma_t\,\sqrt{t}}\right)\,\,.
\end{equation*}
\end{lemma}
}
\begin{proof}
The Lemma can be proven observing that $f_t$ and the rv ${\sigma }_t \, \sqrt{ t} \, z$ in \eqref{eq:EquivInLaw}
have the same characteristic function. The characteristic function of $z$
\begin{equation*}
	\ln \phi^{(z)}(u):= \ln \mathbb{E}\left[e^{ i \, u \, z}\right]=\psi\left(i \, u \, \eta+\frac{u^2 }{2} ; \; k, \; \alpha \right) \, + \,
 i\, u \, \eta\;\;,
\end{equation*}
is s.t. it can be computed using iterated expectations 
 \citep[see, e.g.,][Sec.3.2.2, p.77]{mcneil2015}

\end{proof}

\smallskip

{\it Remark}.
Thanks to Thm.3.1, p.12 of \cite{lukacs1972}, the characteristic function of $z$ is analytic in the horizontal strip $(-\pln,\prn)$ with
\begin{equation}
\label{eq:p_plus_minus}
p^\pm=\lr\eta+\sqrt{\eta^2+2\frac{1-\alpha}{k}}\,\,.
\end{equation}
The proof is analogous to the one in Proposition \ref{pr:Moment}.\\
Thus, the characteristic function of $f_t$ is analytic in the horizontal strip $(-\plt,\prt)$ with
\begin{equation}
p^\pm_t=\frac{1}{\sigma_t\sqrt{t}}\,p^\pm\,\,.
\end{equation}

\smallskip

\begin{proof} {\bf Proposition \ref{th:f_Additive}} 

First, we demonstrate additivity proving that the additive Bachelier satisfies Lemma \ref{lem:f_AdditiveGen}.
In this case, $\sigma_t$ is s.t $\sigma_t^2 t$ is increasing, it goes to zero for $t\rightarrow0^+$ and
\[
\eta_t = \eta, \, k_t = k \;\; .
\]
Condition $1$ of Lemma \ref{lem:f_AdditiveGen} is satisfied because both $p^\pm_t$ and $p^*_t$ are decreasing, while condition $2$ holds. 

Then, we demonstrate that the process is a martingale using Lemma \ref{lem:f_equivalence}
\end{proof}

\smallskip

\begin{proof}  {\bf  Proposition \ref{prop:Lewis} }

The proof reminds the one of \citet{lewis2001} for exponential L\'evy processes, adapted to an underlying described by \eqref{eq:ForwardDynamics}.

Let us define the Fourier transform and its inverse 
\[
\left\{
\begin{array}{rl}
FT: & \hat{\omega}(\xi) : = \displaystyle \int^{+\infty}_{- \infty}\omega(f_t) \, e^{-i \xi f_t}\, \mathrm{d} f_t \\[4mm]
IFT: & \omega(f_t) : =  \displaystyle  \frac{1}{2 \pi} \int^{+\infty}_{- \infty} \hat{\omega}(\xi) \, e^{i \xi f_t}\, \mathrm{d} \xi \,\,,
\end{array}
\right.
\]
the Parseval relation holds
\[
\int^{+\infty}_{- \infty}\mathcal{P}_t(f_t) \, \omega(f_t) \, \mathrm{d} f_t= \frac{1}{2 \pi} \int^{+\infty}_{- \infty} \phi_t(\xi) \, \hat{\omega}(\xi) \,  \mathrm{d} \xi\,\,,
\]
where $\mathcal{P}_t(f_t)$ is the pdf of $f_t$.

\smallskip

\textcolor{black}{Let us consider the call option price \eqref{eq:EU_generic} for
$C(x,t; {\bf p})$.
} Let $x \in \R$, the moneyness.
The two parts of the Fourier transform of the call option payoff $  \omega(f_t):= [f_t  -  x]^+$ are:
\[
\begin{array}{lcll}
 \displaystyle \int^{+\infty}_{- \infty} e^{-i \xi f_t} \,  \mathbbm{1}_{\{f_t > x\}}\, \mathrm{d} f_t  & = &   \displaystyle  \frac{e^{-i \xi x}}{i \, \xi} & \Im(\xi) < 0 \\[4mm]
 \displaystyle \int^{+\infty}_{- \infty} f_t \,  e^{-i \xi f_t} \,  \mathbbm{1}_{\{f_t > x\}} \, \mathrm{d} f_t & = &   \displaystyle \frac{e^{-i \xi x}}{- \xi^2} + x  \frac{e^{-i \xi x}}{i \, \xi} & \Im(\xi) < 0 \,\,.\\
\end{array}
\]

Finally, using the Parseval relation above,
\[
\E_0[f_t  -  x]^+ = \frac{1}{2 \pi} \int^{+\infty}_{- \infty} \phi_t(\xi) \,  \frac{e^{-i \xi x}}{ - \xi^2} \,\mathrm{d} \xi \qquad \Im(\xi) < 0
\]
where the integral path is considered slightly below the real axis.

Because the characteristic function $\phi_t(\xi) $  is  analytic for $\Im(\xi) \in (-\plt, \prt) $, applying the Chauchy theorem, we obtain 
\[
\E_0[f_t  -  x]^+ = \frac{\displaystyle e^{x \, a}}{2 \pi} \int^{+\infty}_{- \infty} \phi_t \left(\xi + {i} \,a  \right) \,  \frac{e^{-i \xi x}}{ \left(i \, \xi -  a \right)^2}\,\mathrm{d} \xi \,\,,
\]
where $a \in(-\plt, 0)$.
In the case we want to consider $a\in[0,\prt)$, we apply the residue Theorem, as in \cite{lee2004option}, Thm.5.1 ($a>0$) and Thm.5.2 ($a=0$, Cauchy principal value).\\
For $a\in(0,\prt)$, 
\[
\E_0[f_t  -  x]^+ = -x+\frac{\displaystyle e^{x \, a}}{2 \pi} \int^{+\infty}_{- \infty} \phi_t \left(\xi + {i} \,a  \right) \,  \frac{e^{-i \xi x}}{ \left(i \, \xi -  a \right)^2}\, \mathrm{d} \xi \,\,,
\]
and for $a = 0$ (where the integral is intended as its principal value)
\[
\E_0[f_t  -  x]^+ = -\frac{x}{2}+\frac{\displaystyle 1}{2 \pi} \int^{+\infty}_{- \infty}\left( -\phi_t \left(\xi\right) \,  \frac{e^{-i \xi x}}{ \xi^2}+\frac{1}{\xi^2}\right)\, \mathrm{d} \xi
\]
\end{proof}

\smallskip

\begin{proof} {\bf Proposition \ref{prop:price}} \\
The Proposition is a consequence of Lemma \ref{lem:f_equivalence} and a straightforward application of the iterated expectation.
A call option is
\[
C(x,t; {\bf p}) = B_0 \, \sigma_t \, \sqrt{ t}  \, \E  \left[  z -  y \right]^+ \;,
\]
thus, using the law of iterated expectations
\[
\E \left[ \, \bullet \, \right] =  \E [  \E [  \, \bullet \, | G ] ]
\]
we get the result

\end{proof}

\smallskip

\begin{proof} {\bf Proposition \ref{pr:Existence} }
We can write the right hand side of equation (\ref{eq:IV eq}) as
\be
\label{eq:BachelierFormula_2}
C_b\left(x, t; \mathcal{I}_t(x)\right) = 
B_0  \,   \sqrt{t} \, c_b\left(  \frac{x}{\sqrt{t}}, \mathcal{I}_t(x) \right)  \, \, ,
\en
where $c_b$ is the normalized Bachelier call price \eqref{eq:BasicFormula}.

\smallskip

The following  properties hold $\forall x \in \R$, for fixed maturity $t$:
\begin{enumerate}[label=(\roman*)]
\item Vega is positive ;
\item
$
 \E_0 \left[ F_t - K \right]^+ > [-x]^+ \; ;
$ 
\item
$
\displaystyle\lim_{\sigma \to+ \infty} c_{b}\left(\displaystyle\frac{x}{\sqrt{t}}, {\sigma}\right) =+ \infty   \; ;
$
\item
$
\displaystyle\lim_{\sigma \to 0^+} c_{b}\left(\displaystyle\frac{x}{\sqrt{t}}, \sigma\right) = \left[-\frac{x}{\sqrt{t}}\right]^+  \; .
$
\end{enumerate}
The first property is a consequence of \eqref{eq:Vega}, the second is due to the Jensen inequality 
\[
   \E_0 \left[ F_t - K \right]^+  >  \left[ \E_0 [F_t - K] \right]^+ =   [-x]^+
\]
and the other two come from direct computation.
The first property implies that $c_{b}\left(\frac{x}{\sqrt{t}}, \sigma\right) $ is strictly increasing in $\sigma$, while the last three that it always exists a solution for (\ref{eq:IV eq}), since the Bachelier call price $C_b$ is obtained multiplying the normalized Bachelier call price $c_b$ by a positive constant.
Thus,  $\forall x \in \R$ and $\forall t \in \R^+$, it exists a unique value for the implied volatility $\mathcal{I}_t(x)$, solution of (\ref{eq:IV eq})
\end{proof}

\smallskip
\begin{proof}{\bf Proposition \ref{prop:IyTimeIndep}}\\
Thanks to the homogeneity property \eqref{eq:hom} of the Bachelier formula applied to equations  \eqref{eq:BachelierFormula} and \eqref{eq:ModelPrice_Time}, 
the IV equation  \eqref{eq:IV eq} becomes the (time independent) equation \eqref{eq:MainIV}.
Thus, its solution ${I}(y)$ does not depend on time
\end{proof}
\smallskip
\begin{proof}{\bf Proposition \ref{pr:Regularity}}\\
We define the function $\mathcal{G}(y,I)$ as
\[
\mathcal{G}(y,I):=c_b(y,I)-\mathbb{E}\left[c_b \left(y+\eta (G-1),\sqrt{G} \right) \right]\;\;.
\]
The IV equation \eqref{eq:MainIV} is equivalent to 
\[
\mathcal{G}(y,I(y))=0\;\;.
\]
Thus, applying the implicit function Theorem \citep[see, e.g.,][Thm.9.18, p.196]{rudin1976principles},  the function $I(y)$ is well-defined and ${\cal C}^1(\mathbb{R})$ if
\begin{equation*}
    \begin{cases}
        &\mathcal{G}(y,I(y))=0\\[3mm]
        &\displaystyle \left. \frac{\partial}{\partial I}\mathcal{G}(y,I)\right|_{I=I(y)}\neq 0
    \end{cases}\quad\forall y\in\mathbb{R}\;\;.
\end{equation*}
The first condition comes from the definition of $I(y)$ (\textcolor{black}{that is well posed thanks to Proposition \ref{prop:IyTimeIndep}}).\\
For the second condition, we recall the expression of the (normalized) Bachelier formula \eqref{eq:BasicFormula}.
Thus,
\[
\frac{\partial}{\partial I}\mathcal{G}(y,I)=\frac{\partial}{\partial I}c_b(y,I)=\varphi\left(-\frac{y}{I}\right)>0\quad\forall y\in\mathbb{R}\;\;,
\]
where the last equality appears in Appendix \ref{app:Bachelier_formula}; this proves that $I(y)\in\mathcal{C}^1(\mathbb{R})$.\\
The derivative of $I(y)$ is given by
\[
I'(y)=-\displaystyle \left. \frac{\frac{\partial\mathcal{G}(y,I)}{\partial y}}{\frac{\partial\mathcal{G}(y,I)}{\partial I}}\right|_{I=I(y)}\;\;,
\]
in this case
\begin{equation}
I'(y)=\frac{\Phi\left(-\frac{y}{I(y)}\right)-\mathbb{E}\left[\Phi \left(-\frac{y+\eta(G-1)}{\sqrt{G}}\right)\right]}{\varphi\left(-\frac{y}{I(y)}\right)}\;\;,
\label{eq:I_first_der}    
\end{equation}
that is well defined since it's the composition of continuous functions of $y$.
With some computations, we can prove that

\begin{equation}
\label{eq:I_second}
    I''(y)= \displaystyle
     \frac{\mathbb{E}\left[\frac{1}{\sqrt{G}} \, \varphi\left(\frac{y+\eta(G-1)}{\sqrt{G}}\right)\right]}{\varphi\left(-\frac{y}{I(y)}\right)}
     -\frac{1}{I(y)}\left(1-\frac{y}{I(y)}I'(y)\right)^2\,\,,
\end{equation}
which is well defined $\forall y\in\mathbb{R},\,\eta\in\mathbb{R},\,k\in\mathbb{R}^+$
\end{proof}
{\it Remark}.
It is possible to compute the derivative of the implied vol $I(y)$ at any order $n$ (where it is well-defined) via the following equation
\[
\frac{d^n}{(dy)^n} c_b \left(y, I(y) \right) = 
(-1)^n  \mathbb{E}\left[\frac{1}{\sqrt{G}} \, He\left(\frac{y+\eta(G-1)}{\sqrt{G}}, n-2\right) \, \varphi\left(\frac{y+\eta(G-1)}{\sqrt{G}}\right)\right]\,\,,
\]
where $He\left( \bullet , n\right)$ are the Hermite polynomials of order $n$ \citep[see, e.g.,][Ch.22]{abramowitz1948handbook}.

\smallskip

\begin{proof} {\bf Proposition \ref{pr:Symmetry}}

For a given $y$, we impose the IV equation \eqref{eq:MainIV}, i.e.
\[
{c_b\left( y, {I}(y;\eta)\right)}  =  \E \left[ c_b\left(y+  \eta \, (G -1) , \sqrt{G}\right) \right]  =: {\cal A}\,\,,
\]
where ${\cal A}$ is a positive value.
We change $y \to (-y)$ and $\eta \to (-\eta)$ in both terms of the IV equation \eqref{eq:MainIV}.
The right hand side becomes
\[
 \E \left[ c_b\left( -(y+  \eta \, (G -1)) , \sqrt{G}\right) \right] =
 y +  \E \left[ c_b\left( y+  \eta \, (G -1) , \sqrt{G}\right) \right] = y + {\cal A}\,\,,
\] 
while the left hand side is
\[
 {c_b\left( -y, {I}(-y;-\eta)\right)}  = y +  {c_b\left( y, {I}(-y;-\eta)\right)}  \;\; .
\]
Thus, we have that
\[
\left\{
\begin{array} {lcl}
{c_b\left( y, {I}(y;\eta)\right)}  &=&  {\cal A} \\
{c_b\left( y, {I}(-y;-\eta)\right)}  & = & {\cal A}  \;\; .
\end{array}
\right.
\]
Because the value of \textcolor{black}{the} implied vol that solves both equations is unique, we have that
\[
I(-y;-\eta) = I(y;\eta)
\]
\end{proof}

\begin{proof} {\bf Corollary \ref{cor:Symmetry} }

The if condition is a consequence of Proposition \ref{pr:Symmetry}
\[
I(y; \eta=0) = I(- y; \eta=0) \; .
\]
Moreover, we can compute $I'_0$, the ATM derivative of $I(y)$ by \eqref{eq:I_first_der},
\begin{equation}
\label{eq:I0_prime} 
I'_0 = 
\displaystyle \left. \frac{\displaystyle \Phi\left( - \frac{y}{I(y)} \right) - \E \left[ \Phi\left( - \frac{y+\eta(G-1)}{ \sqrt{G}} \right) \right] }
{\displaystyle \varphi\left(-\frac{y}{I(y)}\right)} \right|_{y=0} 
 = -\sqrt{\frac{\pi}{2}}\, \E \left[ {\rm erf} \left( \frac{ {\eta}}{\sqrt{2}} \, \frac{  1 - G }{  \sqrt{G}} \right) \right]\,\,,
\end{equation}
that is an odd function of $\eta$, and it is equal to zero only if $\eta =0$. This proves the only if part
\end{proof}

\smallskip

\begin{proof} {\bf Proposition \ref{pr:ATMVolSkew} }

Thanks to Proposition \ref{pr:Regularity}, we can Taylor expand the implied volatility $I(y)$ around the ATM 
\[
{I}(y) = {I}_0 +  I'_0 \, y + \frac{1}{2} I''_0 \, y^2 + o(y^2)\,.
\]

$I_0$ can be obtained from \eqref{eq:MainIV}
observing that $c_b(0, I_0) = \displaystyle\frac{I_0}{\sqrt{2 \pi}}$.

\smallskip

$I'_0$ can be found in \eqref{eq:I0_prime}
\begin{equation*}
I'_0 = -\sqrt{\frac{\pi}{2}}\, \E \left[ {\rm erf} \left( \frac{ {\eta}}{\sqrt{2}} \, \frac{  1 - G }{  \sqrt{G}} \right) \right]\,\,,
\end{equation*}
that is antisymmetric in $\eta$. 
The value of $I''_0$ can be obtained directly 
from the proof of Proposition \ref{pr:Regularity} by (\ref{eq:I_second}) computed ATM
\begin{equation*}
{ I}''_0  =\displaystyle \left. \left(\displaystyle\frac{\mathbb{E}\left[\frac{1}{\sqrt{G}} \, \varphi\left(\frac{y+\eta(G-1)}{\sqrt{G}}\right)\right]}{\varphi\left(-\frac{y}{I(y)}\right)}
     -\frac{1}{I(y)}\left(1-\frac{y}{I(y)}I'(y)\right)^2 \right)\right|_{y=0}= \displaystyle \sqrt{2 \pi}\,\displaystyle 
 \E \left[ \frac{1}{  \sqrt{G}} \, \varphi \left(  {\eta} \, \frac{  1 - G }{  \sqrt{G}} \right) \right] - \frac{1}{I_0}
\end{equation*}
\end{proof}

\smallskip

\begin{proof} {\bf Proposition \ref{pr:Moment}}\\
Because ``the purely imaginary points on the boundary of the strip of
regularity (...) are singular points of" the characteristic function \citep[cf.][Th.3.1, p.12]{lukacs1972}, to select the boundaries of the analyticity strip,
we have to consider $u = i a$ with $a \in \R$, or equivalently for
\[
 - \, a \, \eta\,\sigma_t\sqrt{t} - \frac{ a^2}{2}\sigma_t^2\,t > - \frac{(1-\alpha)}{k}\quad\Leftrightarrow\quad a \in (-\plt, \prt)\,\,,
\]
with 
\[
p^\pm_t =  \frac{1}{\sigma_t\sqrt{t}}\left\{\lr  \eta + \sqrt{ {\eta}^2 + 2 \,  \frac{(1-\alpha)}{k} }\right\}=\frac{1}{\sigma_t\sqrt{t}}\,p^\pm \;\; .
\]
We apply Theorem 3.6 in \cite{smileWings} and obtain that
\[
\lim_{x\rightarrow\pm\infty}\frac{\mathcal{I}_t(x)^2}{|x/\sqrt{t}|}=\frac{1}{2\,p_t^{\lr}\sqrt{t}}\,\,,
\]
because the characteristic function of returns $\zeta_t:=\frac{f_t}{\sqrt{t}}$ is equal to $\phi_t\left(\frac{u}{\sqrt{t}}\right)$ and it is analytic in the horizontal strip $(-\plt\sqrt{t},\prt\sqrt{t})$.\\
Finally, we observe that
\[
\lim_{y\rightarrow\pm\infty}\frac{I(y)^2}{|y|}=\frac{1}{\sigma_t}\lim_{x\rightarrow\pm\infty}\frac{\mathcal{I}_t(x)^2}{|x/\sqrt{t}|}=\frac{1}{2\,p^\lr_t\sigma_t\,\sqrt{t}}=\frac{1}{2\,p^\lr}
\]
\end{proof}

\smallskip

\begin{proof}{\bf Proposition \ref{pr:altEqCalibration}}\\
It is a consequence of the homogeneity property \eqref{eq:hom} of the Bachelier formula \eqref{eq:BachelierFormula} applied to
the IV equation \eqref{eq:MainIV}
\end{proof}

\smallskip

\begin{proof}{\bf Lemma \ref{lem:sticky}}\\
The proposed model is naturally sticky delta, i.e.\ keeping constant the model parameters ${\bf p}$, 
the whole volatility surface as a function of the moneyness does not change
\end{proof}
\newpage
\bibliography{main}
\bibliographystyle{tandfx}

\section*{Notation and shorthands}


	\hspace{-0.6cm}
	\begin{tabular} {|c|l|}
		\toprule
		\textbf{Symbol}& \textbf{Description}\\ \bottomrule
		\textcolor{black}{$(\Omega,\mathcal{F},\mathbb{F},\mathbb{P})$} & \textcolor{black}{filtered probability space} \\
        $\E_s[\bullet]$ & $\E_s[\bullet | {\cal F}_s]$ where ${\cal F}_s$ is the \textcolor{black}{filtration at time $s$} \\
        \textcolor{black}{$(A_s, \nu_s, \gamma_s)$} & \textcolor{black}{Generating triplet characterizing the process $\{f_s\}_{s\geq0}$}\\
		$\mathbbm{1}_A$ & indicator function of $A$, $\mathbbm{1}_A=1$ if $A$ is true, $0$ otherwise\\
		$B_0$ & discount factor between the value date and option maturity $t$ \\
		$F_s$ & forward price with expiry in $t$ and valued at time $s \in [0,t]$ \\  
		$C_b \left(x, t; \sigma \right)$ & Bachelier call option price wrt moneyness $x$ and maturity $t$, with volatility $\sigma$ \\
		$c_b(y,\sigma)$ & normalized Bachelier call option price wrt moneyness $y$, with volatility $\sigma$ \\
		$C \left(x, t; \mathbf{p} \right)$ & additive Bachelier call price wrt moneyness $x$ and ttm $t$, function of parameters $\mathbf{p}$ \\
		$C^{mod} \left(x, t; \eta,k \right)$ & additive Bachelier call price wrt moneyness $x$ and ttm $t$, imposing $\sigma_t=\sigma_t^{ATM}/I_0$ \\
		$C^{mkt} \left(x, t \right)$ & market price of a call option wrt moneyness $x$ and maturity $t$ \\
		$C_s^{mkt} \left(x, t \right)$& market price of a call option wrt moneyness $x$ and maturity $t$, at value date s \\
		$RMSE_{t_0}(\eta^*, k^*)$ & RMSE between market and model prices at value date $t_0$, using parameters $\{\eta^*,k^*\}$\\
		${\sigma}^{ATM}_t$ & ATM vol term structure observed on market data\\
		${\sigma}_t \in \R^+ \setminus 0$ & vol term structure, a model parameter, a strictly positive continuous function\\
		${\eta} \in \R$ & skew parameter \\
		${k} \in \R^+$ & vol-of-vol parameter, a non-negative parameter\\
		$\alpha\in[0,1)$ & hyper-parameter of the model \\
		$\mathcal{I}_t\left(x\right)$ & implied volatility wrt the moneyness $x$ and the maturity $t$\\
            ${ I}\left(y\right)$ & implied volatility wrt the (normalized) moneyness degree $y$ \\
		${\Ib}\left(\chi\right)$ & implied volatility wrt $\chi$   \\
		$G$ & a positive rv with unitary mean and variance ${k}$, Laplace transform $\psi (u; k, \alpha)$ \\
		$\Im(\xi)$ & imaginary part of a complex number $\xi$\\
		$\phi_t\left(u\right)$ & characteristic function of the underlying process $\left\{f_t\right\}_{t\geq0}$\textcolor{black}{, defined in equation \eqref{eq:ChFun}} \\	
		${\varphi} ( \bullet ), {\Phi} ( \bullet )$ & pdf and cdf of a standard normal rv\\
		$\psi (u; k, \alpha)$ &  Laplace exponent of the positive rv $G$, defined in \eqref{eq:laplaceG}\\
		$t_0=0$ & value date \\
		$t$ & European option maturity (and time-to-maturity) \\
		$x$ & option moneyness, equal to $K - F_0$\\	
            	$y$ & $x/(\sigma_t \,\sqrt{ t})$: (normalized) moneyness degree \\	
		$\chi$ & $x /(\sigma^{ATM}_t \, \sqrt{t})$: a normalized  {\it moneyness}\\	
		$\zeta_t$ & $\frac{F_t-F_0}{\sqrt{t}}$: increment of forward price divided by the square root of the ttm $t$\\	
		$\mathcal{C}^n(\mathbb{R})$ & class of $n$-times differentiable functions, with $n$-th derivative continuous \\
		\bottomrule		
	\end{tabular}
	

\bigskip 

	\begin{tabular}{|c|l|}
		\toprule
		\textbf{Symbol}& \textbf{Description}\\ \bottomrule
		\textcolor{black}{a.s.} & \textcolor{black}{almost surely} \\
        ATM & At-The-Money (referred to options) \\
		bp & basis point (typical measurement unit for rates), equal to $0.01\%$ \\
		cdf & cumulative distribution function \\
		CME & Chicago Mercantile Exchange \\
		FT & Fourier transform \\
		ICE & Intercontinental Exchange \\
		iff & if and only if \\
		IFT & inverse Fourier transform \\
		ITM & In-The-Money (referred to options) \\
		IV & implied volatility \\
		mkt & market \\
            	RMSE & root mean square error \\
		NIG & Normal Inverse Gaussian \\
		OIS & Overnight Indexed Swap \\
		OTM & Out-The-Money (referred to options) \\		
		pdf & probability density function \\
		rv & random variable \\
		s.t. & such that \\
         		ttm & time-to-maturity \\
		VG & Variance Gamma \\
		wrt  & with respect to \\
		WTI & West Texas Intermediate oil market \\
		\bottomrule
	\end{tabular}

\end{document}